\documentclass[a4paper,UKenglish,cleveref, autoref, thm-restate]{lipics-v2021}

\title{Positional Properties in Temporal Logic} 

\author{Jessica Newman}{University of Southampton, UK}{jln1g19@soton.ac.uk}{https://orcid.org/0009-0009-8400-0609}{}

\author{Benjamin Plummer}{University of Southampton, UK}{bjp1g19@soton.ac.uk}{https://orcid.org/0009-0007-3069-2676}{}

\authorrunning{J. Newman and B. Plummer} 

\Copyright{Jessica Newman and Benjamin Plummer} 



\category{} 

\acknowledgements{We thank Corina Cîrstea and Enrico Marchioni for their valuable feedback.}

\nolinenumbers 

\usepackage[only,llbracket,rrbracket]{stmaryrd}

\usepackage{wrapfig}
\usepackage{tikz}
\usetikzlibrary{automata, positioning}
\usetikzlibrary{decorations.pathmorphing, positioning}

\EventEditors{Ana Sokolova and Patrick Totzke}
\EventNoEds{2}
\EventLongTitle{37th International Conference on Concurrency Theory (CONCUR 2026)}
\EventShortTitle{CONCUR 2026}
\EventAcronym{CONCUR}
\EventYear{2026}
\EventDate{September 1--4, 2026}
\EventLocation{Liverpool, UK}
\EventLogo{}
\SeriesVolume{391}
\ArticleNo{10}

\begin{document}

\maketitle

\begin{abstract}

We study positional properties in the context of game-based
reactive synthesis.
Our motivation stems from having a usable specification logic,
for which tractable synthesis is guaranteed.
We demonstrate that every $\omega$-regular positional property (with respect to state- or edge-labelled game graphs),
is expressible in linear-time temporal logic.
Additionally, we provide some necessary and sufficient conditions
for when an $\omega$-regular property is positional, and identify well-behaved subclasses
of $\omega$-regular positional properties. Using varieties of languages, we prove that no class of $\omega$-regular positional properties can simultaneously contain a prefix-independent property and be closed under Boolean operations.
We conclude by discussing the implications on alternating-time temporal logic,
where we isolate a few different fragments with tractable model checking,
and compare the associated expressivity of such fragments.

\end{abstract}

\section{Introduction}

Infinite duration, two-player graph games provide a fundamental model of the
interaction between a {\it system} and its {\it environment}. The problem of
{\it reactive synthesis} (generating a system controller which implements a specification)
can be phrased as computing a winning strategy in such a game.
This technique quickly becomes intractable, and is highly sensitive to
the (logical) formalism chosen to express the desired property of the system.
Even when using the property-specification logic linear-time temporal logic ({\sf LTL}) \cite{4567924},
which has powerful automata-theoretic tools,
the synthesis problem is famously 2EXPTIME-complete \cite{Bloem2018}.
This provides a major obstacle to applying synthesis in practical applications.

The difficulty of synthesis is inextricably linked to the vast number of
possible strategies:
on any non-trivial game graph (even with a single player),
there are an infinite number of possible strategies,
because strategies can use an arbitrary amount of memory.
If one could impose a maximum bound on the amount of memory that a strategy can use,
we could reduce the space of strategies to be finite (at least for games
with a finite state space).
This, in turn, reduces {\sf LTL} synthesis to {\sf LTL} {\it model checking} - which is PSPACE-complete \cite{10.1145/3828.3837} - 
by guessing a strategy which uses memory less than the bound, and then performing model checking on each induced transition
system.
Notice that this discussion also carries over to model checking of the
{\it alternating-time temporal logic} {\sf ATL}$^*$ \cite{10.1145/585265.585270} interpreted over game structures,
whereby combining the classical labelling algorithm (see \cite{ALLENEMERSON1987275}) with 
{\sf LTL} model checking the path formula (which performs {\sf LTL} synthesis
in the manner just described)
enables us to also reduce {\sf ATL}* model checking to {\sf LTL} model checking.
However, in both cases
({\sf LTL} synthesis and {\sf ATL}* model checking) this approach is doomed to fail, as the amount of memory that winning
strategies may require is unbounded.


In this work, we investigate to what extent we can syntactically restrict
the collection of {\sf LTL} formulas (or, more generally, the path formula in {\sf ATL}*),
to obtain a class of properties whose winning strategies'
memory can be bounded by a single state - i.e. properties where {\it positional}
strategies suffice.
{\sf LTL} is recognised as having a good balance between {\it expressiveness}
and {\it usability}, and our key result (\Cref{theorem:state_games_positional}) is a positive one:
{\it every $\omega$-regular positional
property is expressible in {\sf LTL}}.
As we shall see, which properties are positional is sensitive to
the kind of games which we are doing synthesis over, specifically to
whether labels (sometimes called colours) appear on the {\it edges} or on the {\it states}.
Our result is with respect to the larger class of positional properties:
those which are positional over {\it state-labelled} games.
These are a natural class of properties in {\sf ATL}$^*$, 
because atomic propositions are assigned to states. 

Recently, (edge-labelled) positional $\omega$-regular objectives have received a full characterisation in \cite{theoretics:14945}, providing different classes of automata which recognise them.
One key observation in loc. cit. is that a positional language $L$ always has a \textit{totally ordered set of residuals} (recall that,
for some finite word $u$, the residual is the set containing all suffixes $v$ such that $uv$ is in $L$). 
This requirement lends itself to an {\it automata-theoretic} perspective, as characterisations can be built on top of the automaton of residuals, which has as states the residuals of the language.
However, the requirement of totally ordered residuals can instead be presented as a purely syntactic closure condition on the language\footnote{if a word $uv\in L$ and a word $xy\in L$, then either $xv\in L$ or $uy\in L$.}. The authors of \cite{theoretics:14945} note that a \textit{language-theoretic} characterisation, using conditions of this form, is left open. We complement their results by providing such a language-theoretic characterisation of positional $\omega$-regular languages over edge-labelled arenas in terms of additional closure conditions on the language (\Cref{cor:omega-pos}). This has two main benefits: firstly, these conditions can be phrased algebraically, enabling the use of tools from algebraic language theory (we pursue this in \Cref{sec:varieties}). Secondly, they naturally extend to the setting of state-labelled positionality (\Cref{prop:state_pos_characterisation}), which is less clear in the automata-theoretic case. We further investigate this in \Cref{prop:edgevsstate}, finding that \textit{a property is positional over state-labelled arenas when it can be played optimally with memory of the previous label in edge-labelled arenas}. 

As mentioned above, we find that LTL can express all $\omega$-regular positional properties. However, this does not tell us much about how to concretely specify such properties. Therefore, in the latter half of the paper we use the language-theoretic characterisation to continue the investigation of specification languages for positional properties. Firstly, we obtain a precise characterisation of positional $\omega$-regular \textit{prefix-independent} languages (\Cref{prop:pre_ind_conditions}), and provide a concrete form for specifying a wide range of such languages. We then investigate whether we can isolate a {\it language variety}
(of infinite words)
consisting entirely of positional properties.
This would give us a good candidate for a logical specification language
to express positional properties, as language varieties
are closed under Boolean operations.
Unfortunately, we present a negative result in \Cref{coro:no_variety_pos_preind}, that:
{\it any (non-trivial) variety of languages consisting only
of positional properties does not contain any prefix-independent properties}.
This immediately rules out the existence of formalisms for expressing positional properties that both support combining positional properties via logical connectives and include even basic positional prefix-independent properties such as Büchi objectives. Our final endeavour
is to present several possible restrictions of the path formula
in {\sf ATL}*, which ensure the resulting fragments can only express positional properties. We isolate these fragments using the previously obtained characterisation, discuss the associated model checking problems, and give a fine-grained analysis
of their expressiveness, comparing them to the logics {\sf EATL} and ${\sf EATL}^+$ \cite{10.1145/567067.567081}, and providing examples of which positional properties can/cannot be expressed in each fragment.

\subparagraph*{Related Work.}
Determining which objectives in infinite games are positional has been well-studied, with a prominent example being that parity objectives are bipositional \cite{185392}. 


 A particularly interesting question is how to characterise which objectives are positional within a given class. Several sufficient conditions have been found for positionality in the class of prefix-independent objectives, including in \cite{10.1007/11787006_29} and later \cite{Bianco2011}. However, a full characterisation of prefix-independent positional objectives remains open. It has been found, however, that the bipositional prefix-independent objectives are all parity objectives \cite{COLCOMBET2006190}. 

 In terms of positional $\omega$-regular objectives, some sufficient conditions were given in \cite{10.1007/978-3-540-74915-8_7}, and a more general characterisation of memory requirements in \cite{theoretics:9608}. Further results on memory requirements for different objectives can be found in \cite{vandenhove:tel-04095220}. As mentioned above, some full characterisations of positional $\omega$-regular languages have been presented in \cite{theoretics:14945}.

Further classes where positionality has been characterised include languages recognised by deterministic Büchi automata \cite{bouyer_et_al:LIPIcs.CONCUR.2022.20}, safety games \cite{lmcs:10547}, and the Borel class $\Sigma^0_2$ \cite{ohlmann_et_al:LIPIcs.STACS.2024.54}. A general characterisation of positionality in terms of universal graphs was given in \cite{theoretics:9724}.

 A full characterisation of bipositional objectives (for finite games) was provided in \cite{10.1007/11539452_33}.

The positionality of model checking has been previously noted for fragments of temporal logic such as {\sf ATL} \cite{10.1007/s10458-013-9231-3}, and there has been some analysis on restricting certain variants of {\sf ATL}$^*$ to positional strategies \cite{SCHOBBENS200482, BUSARD2015128}. It has also been found that all formulae in {\sf ATL}$^+$, a fragment with non-positional path formulae, can be expressed in {\sf ATL}, a positional fragment \cite{10.1007/3-540-47813-2_20}. In a similar spirit, some fragments of Strategy Logic have been identified in which the strategies required are less complex than the full Strategy Logic \cite{10.1007/978-3-319-09764-0_10, gardy_et_al:LIPIcs.STACS.2018.34}.


In a similar spirit to the motivation of our work on temporal logic,
GR(1) synthesis \cite{BLOEM2012911} syntactically restricts LTL specifications to obtain a polynomial-time synthesis algorithm
(we comment on the relationship between our approaches in \Cref{remark:gr1}).


\subparagraph*{Contributions \& Organisation.}

We start, in \Cref{sec:prelim}, by fixing definitions and notation, and recalling basic facts we use about: two-player games, positional properties, language theory, and temporal logic.
In \Cref{sec:edgepos}, we investigate {\it edge-labelled} positional properties, we discuss: expressing positional properties in {\sf LTL} (\Cref{subsec:starfree}); a language-theoretic characterisation of $\omega$-regular positionality (\Cref{subsec:edge_syntactic}); and how this characterisation simplifies in the case of prefix-independent $\omega$-regular languages (\Cref{subsec:edge_pre_ind}).
\Cref{sec:statepos} then relates our findings to {\it state-labelled} positionality: we give a language-theoretic characterisation for this more general class of properties (\Cref{subsec:state_syntactic}), and show how our expressibility result transfers to the state-labelled setting (\Cref{subsec:edge_to_state}).
In \Cref{sec:varieties} we prove our no-go theorem stating that no $\infty$-variety can exist which consists of positional properties and contains a prefix-independent property.
In \Cref{sec:temporallogic}, we describe how one can perform efficient model checking for positional fragments of {\sf ATL}$^*$, and propose two such fragments.
Finally, \Cref{subsec:inexpress} gives various inexpressibility results, establishing a hierarchy including several variants of {\sf ATL}.

We now outline our contributions, we:
\begin{itemize}
    \item Prove that all $\omega$-regular languages
        which are positional over edge-labelled, and state-labelled, games are expressible in {\sf LTL} (\Cref{theorem:positional_express_ltl,theorem:state_games_positional}).
    \item Obtain a language-theoretic characterisation for when an $\omega$-regular language is positional over edge-labelled, and state-labelled, games (\Cref{cor:omega-pos,prop:state_pos_characterisation}). We note the former has also been independently proven by Colcombet and Idir \cite{colcombet2026algebraiccharacterisationevepositionallanguages}. Furthermore, we show this can be simplified when the language is prefix-independent (\Cref{prop:pre_ind_conditions}).
    \item Provide a polynomial-time decision procedure to decide whether a Wilke algebra recognises a positional language (\Cref{coro:wilke_algebra_recog}).
    \item Give a concrete description of a large class of positional prefix-independent languages (\Cref{prop:pref_ind_pos_concrete}).
    \item Show that there is no language variety consisting entirely of positional languages which contains a prefix-independent language (\Cref{coro:no_variety_pos_preind}).
    \item Establish that the {\sf ATL}$^*$ model checking problem (which can state {\sf LTL} synthesis) for positional fragments is in PSPACE, and for bipositional fragments is in $\Sigma^p_2$
    (\Cref{prop:synthesiscomplexity}).
    \item Present several positional fragments of {\sf ATL}* and establish an expressivity hierarchy among them, and other variants of {\sf ATL} (\Cref{theorem:ATL_fragments}).
\end{itemize}

\section{Preliminaries}\label{sec:prelim}


\subparagraph{Words.} Fix an alphabet $\Sigma$.
We use $\Sigma^*$, $\Sigma^+$, $\Sigma^\omega$, and $\Sigma^\infty$ to denote the set finite words, finite non-empty words, infinite words, and possibly infinite words (i.e. $\Sigma^\infty:=\Sigma^\omega\cup\Sigma^*$), respectively.
Note we can take a relation as the alphabet, and then words will be sequences of pairs in the relation.
We write $w[i..j]$ for sub-word of $w$ from index $i$ to index $j$ inclusive, for any word $w\in\Sigma^\infty$ which has length greater than or equal to $i$ and $j$. We write $w[i]:=w[i,i]$.

\subparagraph{$\omega$-regular languages.} $\omega$-regular languages can be equivalently defined, inter alia, as the languages recognised by finite $\omega$-semigroups and finite Büchi automata \cite{carton:hal-00340797}, expressible in the monadic second-order logic S1S \cite{RICHARDBUCHI19661}, and given by $\omega$-regular expressions \cite{MCNAUGHTON1966521}. They include languages expressible in {\sf LTL} \cite{WOLPER198372}. Two $\omega$-regular languages are equivalent exactly when they have the same set of ultimately periodic words, i.e. words of the form $uv^\omega$ for $u,v\in\Sigma^*$ \cite{10.1007/3-540-58027-1_27}. Given a language $L\subseteq \Sigma^\omega$ and word $u\in \Sigma^*$, we can obtain the residual $R_L(u) = \{v\in \Sigma^\omega\mid uv\in L\}$. For $\omega$-regular languages, the set $R(L) = \{R_L(u)\mid u\in \Sigma^*\}$ of residuals is finite \cite{MALER199793}. We say a language $L$ has totally ordered residuals when $R(L)$ is totally ordered by inclusion. A language $L$ is \textit{prefix-independent} when for any word in the language, we can add or remove any finite prefix and remain in the language, i.e. for all $x\in \Sigma^\omega$, for any $u\in \Sigma^*$, $x\in L$ iff $ux\in L$.

\subparagraph{Arenas.} Two-player games are played on {\it arenas}: tuples $(V_1,V_2,\delta)$ where $V_1$ and $V_2$ are finite sets of states controlled by Player 1 and Player 2 respectively, and $\delta\subseteq V\times V$ is a total relation, where we write $V$ for the disjoint union of $V_1$ and $V_2$.
Given an edge $e\in\delta$ where $e=(v,v')$, we define $src(e):=v$ and $tgt(e):=v'$.
Our arenas will implicitly come equipped with a {\it labelling}, where labels comes from a set $\Sigma$.
These come in two forms: {\it edge-labelled} and {\it state-labelled}. A \textit{state-labelled arena} comes equipped with a function $\pi:V \rightarrow \Sigma$, assigning labels to states. An \textit{edge-labelled arena} replaces the above $\delta$ with a total relation $\delta\subseteq V \times(\Sigma\times V)$, associating each edge with a label, note there can now be multiple edges between states (but with different labels). Given an element $e\in\delta$ such that $e=(v,c,v')$, we define $src(e):= v$, $tgt(e) := v'$ and $col(e) := c$. When drawing diagrams of arenas, we will sometimes draw an edge $x\to y$ labelled with a word $a_0\dots a_n\in\Sigma^+$, when it determines a unique sequence of transitions $x\xrightarrow{a_0}\cdots\xrightarrow{a_n}y$ in the arena.

\subparagraph{Games \& Strategies.} Given a (labelled) arena $(V_1,V_2,\delta)$, a \textit{strategy} is a function $\sigma: V\cup\delta^+\rightarrow\delta$ such that $src(\sigma(v))=v$ and $\sigma(v)\in\delta$ for all $v\in V$,
and $src(\sigma(we))=tgt(e)$ and $\sigma(we)\in\delta$ for all $w\in\delta^*$ and $e\in\delta$. A strategy selects for each state an edge which is accessible from that state, and for each nonempty sequence of edges an edge which is accessible from the target state of the final edge in that sequence.
We define a {\it positional strategy} as a function $\sigma: V \rightarrow\delta$,
such that $src(\sigma(v))=v$ and $\sigma(v)\in\delta$ for all $v\in V$.
Note that our definitions of strategy are a little nonstandard, as they pick moves at states controlled by both players.
The players are differentiated in the following definition of {\it play}.

\begin{wrapfigure}[3]{4}{9.6cm}
\vspace{-1em}
\begin{tabular}{cc}
    $\varsigma[0] = \begin{cases} \sigma_1(q) & \text{if $q\in V_1$} \\ \sigma_2(q) & \text{if $q\in V_2$}\end{cases}$
    & $\varsigma[i+1] = \begin{cases} \sigma_1(\varsigma[0,i]) & \text{if $tgt(\varsigma[i])\in V_1$} \\ \sigma_2(\varsigma[0,i]) & \text{if $tgt(\varsigma[i])\in V_2$}\end{cases}$
\end{tabular}
\end{wrapfigure}
Given an arena $(V_1,V_2,\delta)$, an initial state $q\in V$, and two strategies $\sigma_1,\sigma_2$, a \textit{play} $\varsigma\in \delta^\omega$ is defined inductively on the right.
The definition is analogous for positional strategies, but we only require the previous state $tgt(\varsigma[i])$ rather than the play so far. Every play $\varsigma$ generates an associated {\it trace} $\lambda\in \Sigma^\omega$. For state-labelled arenas, we can set $\lambda[i] = \pi (src(\varsigma[i]))$. For edge-labelled arenas, we can set $\lambda[i] = col(\varsigma[i])$). We will denote the play generated by strategies $\sigma_1,\sigma_2$ as $play(\sigma_1,\sigma_2,q)$ and the trace as $trace(\sigma_1,\sigma_2,q)$.
 Given a labelled arena, a strategy $\sigma$, and an initial state $q$, we define the {\it Player 1 outcome of $\sigma$} as the set $out_1(\sigma,q):=\{trace(\sigma,\sigma',q)\mid \sigma'\text{ is a strategy}\}$. The {\it Player 2 outcome of $\sigma$} $out_2(\sigma,q)$ can be defined analogously (this time the first component in $trace$ will be quantified over).

A {\it property} is a subset of traces $L\subseteq \Sigma^\omega$. We can assign a property as the \textit{winning condition}, or \textit{objective}, of a game played over an arena. Given an arena with winning condition $L$ and initial state $q$, we say Player 1 {\it is (positionally) winning} iff there is a (positional) strategy $\sigma_1$ such that $out_1(\sigma_1, q)\subseteq L$. Player 2 is {\it (positionally) winning} iff they have a (positional) strategy $\sigma_2$ such that $out_2(\sigma_2, q)\subseteq\Sigma^\omega\setminus L$.
It is well known that $\omega$-regular languages are \textit{determined} winning conditions \cite{Bchi1969}: Player 1 is winning iff Player 2 is not winning.

\subparagraph{Positionality.}
A property $L$ is {\it state-labelled (edge-labelled) positional}, if whenever Player 1 has a strategy $\sigma$ in some state-labelled (edge-labelled) arena with objective $L$ where $\sigma$ is winning for P1 in all states $w\in W$ for some $W\subseteq V$, there exists a positional stategy $\sigma'$ such that P1 winning from all $w\in W$ also.
A property $L$ is \textit{state-labelled (edge-labelled) bipositional} if both $L$ and $\Sigma^\omega\setminus L$ are state-labelled (edge-labelled) positional. 
We can restrict the notions of positionality to arenas where all states are controlled by Player 1 (where $V_2=\emptyset$).
The resulting notions are {\it 1P edge-labelled positional} and {\it 1P state-labelled positional} properties.
For positionality restricted to 1P arenas, we differentiate between \textit{uniform} and \textit{non-uniform} positionality: uniform positionality is as defined above, whereas a property is {\it non-uniform state-labelled (edge-labelled) positional} if whenever Player 1 is winning at a state $q$ in some state-labelled (edge-labelled) arena with objective $L$, they are also positionally winning from $q$. It can be seen that uniform positionality implies non-uniform positionality, and additionally for positionality over 2P arenas this distinction makes no difference, so either definition can be taken. 
When some of these terms are omitted and not clear from context, we take the default to be \textit{uniform} positionality, and to be over 2P arenas.

\subparagraph{Memory.} Strategies $\sigma:V\cup \delta^+\rightarrow V$ as defined above can be implemented with \textit{memory structures}. A memory structure is a deterministic automaton $\mathcal{M}=(M,m_0,\mu)$ where $M$ is a set of memory states, $m_0\in M$ is an initial state, and $\mu: M\times \Sigma$ is an update function, paired with a strategy $\sigma: V \times M \rightarrow \delta$ with the same constraints on chosen edges as the strategies above. When generating a play, the move chosen by the strategy is $\sigma(v,m)$ for the current state $v$ and current memory state $m$. Every time a label $c\in\Sigma$ is seen, the current memory state is updated from $m$ to $\mu(m,c)$. In other words, given a partial play $\varsigma[0,i]$ ending on a state in $V_1$, and a corresponding trace $\lambda[0,i]$, $\varsigma[i+1] = \sigma(tgt(\varsigma[i]),\mu^+(m_0,\lambda[0,i]))$ where $\mu^+$ is the extension of $\mu$ to finite sequences of labels (i.e. for $w\in\Sigma^+$ $\mu^+(m,w) = \mu(\ldots\mu(\mu(m,w[0]),w[1])\ldots w[|w|-1])$). Winning a game with an $\omega$-regular objective requires only a finite-memory strategy \cite{Bchi1969}. Positional strategies are those implemented by memory structures where $|M| = 1$.


\subparagraph{Algebraic Language Theory.} We recall standard notions from algebraic language theory (full details can be found in e.g. \cite{200475}). An \textit{$\omega$-semigroup} is a pair $S = (S_+,S_\omega)$ of a semigroup $S_+$ and a set $S_\omega$, equipped with a \textit{mixed product} $S_+\times S_\omega\rightarrow S_\omega$ and \textit{infinite product} $\pi:S_+^\omega\rightarrow S_\omega$. We say $S$ \textit{recognises} a language $L\subseteq \Sigma^\infty$ if there is a homomorphism $\varphi:\Sigma^\infty\rightarrow S$ from the free $\omega$-semigroup $\Sigma^\infty = (\Sigma^*,\Sigma^\omega)$ to $S$ such that for some $F\subseteq S_+\cup S_\omega$, $\varphi^{-1}(F) = L$. Finite $\omega$-semigroups recognise exactly the ($\omega$-)regular languages, and admit finite presentations as \textit{Wilke algebras}, in which the infinite product is determined by a map $(-)^\omega:S_+\rightarrow S_\omega$.

A \textit{variety} of $\omega$-semigroups is a class of $\omega$-semigroups closed under subsemigroups, quotients, and finite products. An \textit{$\infty$-variety} of languages is a class of recognisable languages closed under Boolean operations, residuals, and preimages of morphisms between free $\omega$-semigroups. Varieties of $\omega$-semigroups are in bijective correspondence with $\infty$-varieties of languages \cite{2004265}. 

\subparagraph{Linear-time Temporal Logic. } {\sf LTL} is a property-specification logic, used to specify subsets of infinite words.
We define the syntax of {\sf LTL} over our set of labels $\Sigma$, by taking the smallest collection of formulas such that each $a\in\Sigma$ is a formula, and $\neg\varphi$, $\varphi\lor\psi$, $X\varphi$, and $\varphi U\psi$ are all formulas when $\varphi$ and
$\psi$ are formulas. The semantics of {\sf LTL} are subsets of infinite words $\Sigma^\omega$, defined in the standard way (see, e.g. \cite[Chapter II, 6]{Demri_Goranko_Lange_2016}).
A key fact we shall use in \Cref{sec:edgepos}, is that {\sf LTL} (when the semantics is taken over possibly infinite words) expresses exactly the {\it star-free languages} \cite{DBLP:conf/birthday/DiekertG08}: those recognisable by finite aperiodic monoids. We shall use {\it $\omega$-star-free} to mean the part of the class which consists of languages of infinite words.

It can be useful to consider the semantics of {\sf LTL} over 1P state-labelled arenas, where we take $\Sigma$ to be $2^{Prop}$, the set of subsets of atomic propositions (these are known as {\it serial Kripke frames} to modal logicians).
Here there are two choices, the {\it existential} or {\it universal} semantics: depending on whether we are looking for a single path through the transition system which generates the formula as a trace, or we demand that every possible trace through the transition system satisfies the formula.

\section{Edge-Labelled Positional Omega-Regular Languages}\label{sec:edgepos}


We refer to edge-labelled positional properties as positional properties in this section.

\subsection{Expressibility in {\sf LTL}}\label{subsec:starfree}

We will prove the statement that all $\omega$-regular positional properties are $\omega$-star-free (or equivalently, expressible as {\sf LTL} formulae).
Our approach is to characterise the $\omega$-star-free languages as those which are {\it counter-free} and $\omega$-{\it counter-free} (we define these notions shortly).
The benefit of this, is that being ($\omega$-)counter-free is a language-theoretic property, which gives us direct methods to show that they follow from a language being positional (in \Cref{cflemma} and \Cref{necwcf}).


\begin{definition}
    A language  $L\in \Sigma^\omega$ is \begin{itemize}
        \item (CF) counter-free whenever $uv^nw\in L$ iff $uv^{n+m}w\in L$ for all $m\in\mathbb{N}$,
        \item ($\omega$-CF) \textit{$\omega$-counter-free} whenever $u(vx^ny)^\omega\in L$ iff $u(vx^{n+m}y)^\omega\in L$ for all $m\in\mathbb{N}$,
    \end{itemize}
    for some $n\in \mathbb{N}$ and all $u,v,x,y\in \Sigma^*$ and $w \in \Sigma^\omega$.
\end{definition}

A language being counter-free enforces that it behaves aperiodically on the finite prefixes of infinite words, whereas it being $\omega$-counter-free enforces that it behaves aperiodically on the factors of ultimately periodic words.
It is well know that $\omega$-star-free languages are those which can be recognised by finite aperiodic monoids \cite[Theorem 1.1]{DBLP:conf/birthday/DiekertG08}, which we can show enforce the same aperiodicity conditions.
Note that the only if direction of \Cref{starfree} relies on the Arnold congruence from \cite{ARNOLD1985333}.

\begin{restatable}{proposition}{starfree}
    \label{starfree}
    An $\omega$-regular language $L\subseteq \Sigma^\omega$ is $\omega$-star-free iff it is counter-free and $\omega$-counter-free.
\end{restatable}

Recall that the set of residuals for a language is: totally ordered when the language is positional \cite[Lemma 4.1]{theoretics:14945}; and finite when the language is $\omega$-regular. These two properties are enough to show that a language is counter-free. Roughly, if residuals are totally ordered by inclusion, then the residuals of any word $uv^i$ form a monotone chain, so membership of a word $uv^iw$ in L can change at most once; otherwise two residuals would be incomparable. But a non-counter-free language forces words where this change happens at arbitrarily high $i$, yielding arbitrarily many distinct residuals, contradicting finiteness.

\begin{restatable}{lemma}{residualcounterfree}
    \label{cflemma}
    If a language $L\subseteq \Sigma^\omega$ has a finite set of residuals $R(L)$ which is totally ordered by inclusion, then it is counter-free.
\end{restatable}

We have a similar result for $\omega$-counter-freeness, by first showing that for each word $u(vw^ix)$, there is a bound on $i$ at which the word no longer changes membership of $L$, and then showing this bound is uniform using the syntactic congruence from \cite{ARNOLD1985333}.

\begin{restatable}{lemma}{regposcounterfree}
    \label{necwcf}
    If a language $L\subseteq\Sigma^\omega$ is $\omega$-regular and positional, then it is $\omega$-counter-free.
\end{restatable}

Combining \Cref{starfree}, \Cref{cflemma}, and \Cref{necwcf}, we obtain:

\begin{theorem}
    \label{theorem:positional_express_ltl}
    If a language $L\subseteq \Sigma^\omega$ is $\omega$-regular and positional, then it is $\omega$-star-free.
\end{theorem}

Since ${\sf LTL}$ can express all $\omega$-star-free languages (and vice versa), then all positional properties can be expressed in ${\sf LTL}$.

\subsection{Language-Theoretic Characterisation}
\label{subsec:edge_syntactic}

In this section we will provide a characterisation for $\omega$-regular properties which are positional over edge-labelled arenas. Unless otherwise stated, in this section we will consider arenas to be finite-state. We shall use the fact that for a game with an $\omega$-regular winning condition, a player has a winning strategy iff they have a winning strategy with memory that can be encoded in finite states \cite{Bchi1969}. In a 1P arena, this will generate an ultimately periodic play, i.e. a play of the form $pl^\omega$ for some $p\in \delta^*,l\in\delta^+$. The main idea is to enforce closure conditions on $L$ which guarantee that for any winning ultimately periodic play $pl^\omega$, there exists a winning ultimately periodic play $p'l'^\omega$ where $p'l'$ contains strictly fewer violations of positionality - choices of different edges from same state - than $pl$;  this is done with closure conditions that guarantee for any cycle in the play beginning and ending on the same state, either we can remove this cycle or take this cycle forever and still have a winning play. Since $pl$ is finite, it contains a finite number of violations of positionality, so we end up with an ultimately periodic play in which any given state is always followed by the same edge. From this we can construct a positional strategy which generates the same play. We can find some closure conditions which achieve this and are also necessary for positionality.
This method lets us characterise {\it non-uniform} 1P positional properties:



\begin{restatable}{proposition}{nonuniformpositional}
    \label{prop:1pnonuniformpos}
    An $\omega$-regular language $L$ is non-uniform 1P positional\footnote{Note the results from the previous section do not necessarily apply to non-uniform 1P positional languages, as this class of languages does not require totally ordered residuals.} iff for all $u,y\in \Sigma^*,$ $v,w,x\in \Sigma^+$:
    \begin{enumerate}
        \item[(1)] if $uvwx^\omega\in L$ then either $uv^\omega \in L$ or $uwx^\omega\in L$
        \item[(2)] if $u(vwy)^\omega\in L$ then either $uvw^\omega \in L$ or $u(vy)^\omega\in L$
    \end{enumerate}
\end{restatable}

The closure under union of certain classes of positional language has been a notable question since \cite{kopthesis}. It can be seen that when we take a finite union of such languages, the resulting language will still satisfy the closure conditions (1) and (2). Therefore:

\begin{restatable}{proposition}{finiteunion}
    $\omega$-regular non-uniform 1P positional languages are closed under finite union.
\end{restatable}

However, there are languages in this class which are not positional on two-player arenas. For example, the property $L= \{a^\omega, b^\omega\}$, which says that either we only see $a$ or we only see $b$; the residuals $R_L(a) = \{a^\omega\}$ and $R_L(b) = \{b^\omega\}$ are incomparable. For games without a given starting state, we require that the language has totally ordered residuals (see \Cref{subsec:starfree}). A finite-memory strategy now generates an ultimately periodic play from each state - this means we can have violations of positionality `across' the plays. In conjunction with the method from \Cref{prop:1pnonuniformpos}, these can be amended using totally ordered residuals, which given two partial plays ending at the same state allows us to replace the continuation of one play with the other.

\begin{restatable}{proposition}{poneuniformpos}
\label{prop:1puniformpos}
    An $\omega$-regular language $L$ is uniform 1P positional iff it has totally ordered residuals and satisfies (1) and (2) from \Cref{prop:1pnonuniformpos}.
\end{restatable}

In \cite{theoretics:14945}, it was shown that uniform 1P positional $\omega$-regular languages are exactly the positional $\omega$-regular languages.

\begin{proposition}
    [\cite{theoretics:14945}] An $\omega$-regular language is positional over finite Player 1 controlled edge-labelled arenas iff it is positional over all edge-labelled arenas.
\end{proposition}

So from this we have our main result for this section, that the same conditions hold for positionality over 2P arenas. 

\begin{corollary}\label{cor:omega-pos}
    An $\omega$-regular language $L$ is positional over 2P edge-labelled arenas (finite or infinite) iff it has totally ordered residuals and satisfies (1) and (2)
    from \Cref{prop:1pnonuniformpos}.
\end{corollary}

We note that a very similar form of the above characterisation has been independently discovered by Colcombet and Idir \cite{colcombet2026algebraiccharacterisationevepositionallanguages}.

We will discuss some corollaries of this characterisation. Firstly, these conditions allow us to give a decision procedure for checking whether an $\omega$-regular language is positional. Given a finite algebraic object recognising a $\omega$-regular language, we can easily check the above conditions from the multiplication table. Therefore:

\begin{restatable}{corollary}{wilkealg}
    \label{coro:wilke_algebra_recog}
    Given a Wilke algebra $S = (S_+,S_\omega)$ recognising a language $L\subseteq \Sigma^\omega$, we can determine whether $L$ is positional in polynomial time.
\end{restatable}

Unfortunately, unlike in the non-uniform 1P case, languages positional over 2P arenas are not closed under finite union in general, as the union of two languages with totally ordered residuals may not necessarily have totally ordered residuals. However, we can take a finite union in restricted circumstances:

\begin{corollary}\label{cor:edgeunion}
    Given a finite set $\mathbf{L}$ of positional $\omega$-regular languages, the union $\cup \mathbf{L}$ is positional iff $\cup \mathbf{L}$ has totally ordered set of residuals.
\end{corollary}

Prefix-independent $\omega$-regular languages are such that $R(L) = \{L\}$ \cite{ANGLUIN2021104598}. Therefore, if we have an arbitrary $\omega-$regular language $L_1$ and a prefix-independent $\omega$-regular language $L_2$, then $R(L_1\cup L_2) = \{X \cup L_2 | X\in R(L_1)\}$. Therefore, if $R(L_1)$ is totally ordered then so is $R(L_1\cup L_2)$, so a special case of the above corollary is the (known \cite{theoretics:14945}) fact that the union of a positional $\omega$-regular language and prefix-independent $\omega$-regular language will be positional.

\subsection{Prefix-Independent Languages}
\label{subsec:edge_pre_ind}

The class of prefix-independent languages, where only the infinite behaviour of a word affects membership, is important in verification. We provide a condition characterising the positional $\omega$-regular prefix-independent languages, and follow with a concrete description of a class of such languages. Applying our characterisation in \Cref{cor:omega-pos}, we find that prefix-independent languages trivially meet many of the conditions:


\begin{restatable}{proposition}{preindconds}
    \label{prop:pre_ind_conditions}
    A prefix-independent $\omega$-regular language $L$ is positional iff $(uv)^\omega\in L$ implies either $u^\omega\in L$ or $v^\omega\in L$.
\end{restatable}

Note that this provides an alternate justification of the known result \cite[Theorem 3.4]{theoretics:14945}
that the class of prefix-independent $\omega$-regular positional languages is closed under finite union.
The condition in \Cref{prop:pre_ind_conditions} is similar in form to \textit{concavity}, a known sufficient condition for positionality of prefix-independent languages \cite{kopthesis}:

\begin{definition}
    [Concavity] For all $x,y\in \Sigma^\omega$ such that $x=x_1x_2\ldots$, $y = y_1y_2\ldots$ and for all $i$, $x_i,y_i\in \Sigma^+$, if $x_1y_1x_2y_2\ldots \in L$ then either $x\in L$ or $y\in L$.
\end{definition}

However, concavity is not necessary for positionality, even for $\omega$-regular prefix-independent languages, as can be seen by the following example:

\begin{example}\label{eg:non_concave}
     $L = \{w\in \Sigma^\omega| \text{  substring $bb$ occurs finitely often in $w$}\}$ is prefix-independent, positional, but not concave, as $(ab)^\omega\in L$ but neither $(aabb)^\omega$ nor $(bbaa)^\omega$ are.
\end{example}

So, we can see the condition in \Cref{prop:pre_ind_conditions} as a weakening of concavity to be both necessary and sufficient for positionality. We will use this characterisation to capture a large class of positional prefix-independent $\omega$-regular languages. An $\omega$-regular language is prefix-independent iff it can be expressed in the form $\Sigma^*(R_1)^\omega + \ldots + \Sigma^*(R_n)^\omega$ \cite{ANGLUIN2021104598}. Since prefix-independent positional languages are closed under finite union \cite{theoretics:14945}, we will focus on languages of the form $\Sigma^*(R)^\omega$. The following form can express many interesting properties:

\begin{restatable}{proposition}{specialform}
    \label{prop:pref_ind_pos_concrete}
    Suppose a language $L$ can be expressed in the form $\Sigma^*(SR)^\omega$, where $S\subseteq \Sigma$, and $R\subseteq (\Sigma\setminus S)^*$ is a subword-closed regular language with totally-ordered residuals. Then, $L$ is prefix-independent and positional.
\end{restatable}

This encodes a condition with both `liveness' and `safety', in the sense that $S$ gives a set of labels which are seen infinitely often and $R$ restricts the behaviour between these labels.  We will provide some examples of languages in this class. In fact, Example \ref{eg:non_concave} above can be expressed in this form, where $S=\Sigma \setminus \{b\}$ and $R = \{b, \varepsilon\}$. Let us explain what is meant by $R$ being \textit{subword-closed} \cite{ATMINAS2024114595}. The subword ordering on a language $L$ is defined such that $u\leq_{SUB}v$ iff there is an increasing sequence $k_1,k_2, \ldots k_{|u|}$ such that $v[k_1]v[k_2]\ldots v[k_{|u|}] = u$ - i.e. all the characters of $u$ can be found within $v$ in the correct order. We can specify a subword-closed language by giving the set of minimal disallowed subwords - a word is in the language iff it does not contain any of these as a subword. The set of disallowed subwords is called the \textit{anti-dictionary}. The requirement for positionality places some restrictions on the entries in the anti-dictionary:

\begin{restatable}{proposition}{ordantidict}
    \label{prop:totally_ordered_anti_dict}
    Suppose $L\subseteq\Sigma^*$ is a subword-closed language with anti-dictionary $D$. $L$ has totally-ordered residuals iff for every $xy,uv\in D$, where $x,y,u,v\in \Sigma^*$, there is some $d\in D$ such that $d\leq_{SUB} xv$ or $d \leq_{SUB} uy$.
\end{restatable}



With this we can define languages where we disallow a single specific subword:

\begin{restatable}{proposition}{subwordclosed}
    Every subword-closed language with a single element anti-dictionary has totally-ordered residuals. 
\end{restatable}

And we can define languages where we disallow a set of labels:

\begin{restatable}{proposition}{antidict}
    If $D$ is the anti-dictionary of a subword-closed language with totally-ordered residuals, then $D\cup\{a\}$ for any $a\in \Sigma$ is the anti-dictionary of a subword-closed language with totally-ordered residuals.
\end{restatable}

A Rabin condition is a finite collection of pairs $(U_i,V_i)$, where $U_i$ must be visited infinitely often, and $V_i$ must be visited finitely often. Since we can encode each $V_i$ as single-character entries in an anti-dictionary, we can represent these conditions as finite unions of languages of the form \Cref{prop:pref_ind_pos_concrete} - this also subsumes parity conditions. Some further examples of interest of subword-closed languages with totally ordered residuals are found in \cite{refId0}.

\section{State-Labelled Positional Omega-Regular Languages}\label{sec:statepos}


\subsection{Language-Theoretic Characterisation}
\label{subsec:state_syntactic}

We can characterise state-labelled $\omega$-regular positional properties with similar conditions to those that characterise edge-labelled positional properties. Recall that a strategy is not positional whenever it visits the same state twice and prescribes different moves each time. In state-labelled arenas, since each state has a single label, we know that if we visit the same state twice the label seen will be the same each time. So, our closure conditions only need to apply when we see the same label twice, as otherwise we have not seen the same state twice.

\begin{proposition}\label{prop:statepos}
    An $\omega$-regular language $L$ is non-uniform 1P state-labelled positional iff the following conditions hold for all $u,y\in \Sigma^*$, $v,w,x\in \Sigma^+$ where $v[0] = w[0]$:
    \begin{enumerate}
        \item[(1)]  if $uvwx^\omega\in L$ then either $uv^\omega \in L$ or $uwx^\omega\in L$
        \item[(2)]  if $u(vwy)^\omega\in L$ then either $uvw^\omega \in L$ or $u(wy)^\omega\in L$
    \end{enumerate}
\end{proposition}

This can be shown with the same method as Proposition \ref{prop:1pnonuniformpos}. These weaker requirements for positionality allow for more properties to be positional. For example, we can have properties that contain permutations of distinct labels, since violations of positionality can only occur when we see the same label twice.

\begin{example}
    The language $(abc)^\omega$ is 1P state-labelled positional.  
\end{example}

In fact, the language in this example is 1P uniform positional. Despite this, the language does not have a totally ordered set of residuals. This requirement is also weaker for state-positional languages; we only require that for words ending in the same label, the set of residuals is totally ordered.

\begin{restatable}{proposition}{totordresiduals}
    If $L$ is a state-labelled positional $\omega$-regular language, then for each $a\in \Sigma$, the set of residuals for words ending in $a$, $\{R_L(ua)| u\in \Sigma^*\}$, is totally ordered by inclusion.  
\end{restatable}

We can now make the analogous statement of \Cref{prop:1puniformpos}, which can be proven with the same method.

\begin{proposition}
    $L$ is uniform 1P state-labelled positional iff for each $a\in \Sigma$, the set $\{R_L(wa) | w\in \Sigma^* \}$ is totally ordered by inclusion, and conditions (1) and (2) hold from \Cref{prop:statepos}.
\end{proposition}

Additionally, as seen previously in the edge-labelled case, the conditions for 1P uniform positionality also characterise positionality on 2P arenas (often called a 1-to-2 player lift), although demonstrating this is quite involved: our plays now form trees, so we have to show these conditions allow us to replace trees of winning plays with trees of winning plays containing fewer violations of positionality. Interestingly, the requirement for totally ordered residuals does most of the `work' here, allowing us to, given two subtrees of plays rooted at the same state, replace one with the other. The other conditions enforce a similar property to progress consistency in \cite{theoretics:14945}; once we have shown we can recursively nest a subtree within itself an arbitrary number of times and still have a winning play, these allow us to extend this nesting to an infinite depth and still have a winning play.

\begin{restatable}{theorem}{stateposchar}\label{prop:state_pos_characterisation}
    $L$ is positional over finite 2P state-labelled arenas iff for each $a\in \Sigma$, the set $\{R_L(wa) | w\in \Sigma^* \}$ is totally ordered by inclusion, and conditions (1) and (2) hold from \Cref{prop:statepos}. 
\end{restatable}

\subsection{Relation to Edge-Labelled Positionality}
\label{subsec:edge_to_state}

There is a connection between edge-labelled positional and state-labelled positional objectives. It is known that edge-labelled positional languages can be played positionally on state-labelled arenas, and that the converse is not true - there are more state-labelled positional objectives than edge-labelled \cite{kopthesis}. We make precise the converse relation between them:

\begin{restatable}{proposition}{stateedgepos}\label{prop:edgevsstate}
    An $\omega$-regular language $L$ is positional over state-labelled arenas iff it can be played optimally with memory of the previous label in edge-labelled arenas.
\end{restatable}
Here, `memory of the previous label' means a strategy $\sigma:V\times \Sigma \rightarrow V$, where when generating a play, the current state $s$ and label of the previous edge $c$ is given to $\sigma$ to obtain the successor state $\sigma(s,c)$.
This can be seen in the extension of the conditions for edge-labelled positionality (\Cref{cor:omega-pos}) to state-labelled positionality (\Cref{prop:state_pos_characterisation}) - since we can `see' the previous label in state-labelled arenas, we only need that these conditions hold at points where the previous labels were equal. \Cref{prop:pre_ind_conditions} can be extended to state-labelled positionality in the same way, giving us that the class of prefix-independent state-labelled positional languages is closed under finite union.

Similarly to the edge-labelled case, state-labelled positional $\omega$-regular languages are expressible in {\sf LTL}. This can be shown through a very similar argument as in Section \ref{subsec:starfree}. Note that the proof of Lemma \ref{cflemma} also holds for languages where just the sets $\{R_L(ua) | u\in \Sigma^*\}$ are totally ordered by inclusion for each $a\in \Sigma$, since $uv^i$ will end in the same character as $uv^j$ for any $i,j >0$. For the analogue to Lemma \ref{necwcf}, it is simple to adapt the given counterexample to a state-labelled game, from which the same argument follows. Therefore:

\begin{theorem}
    \label{theorem:state_games_positional}
    If a language $L\subseteq \Sigma^\omega$ is $\omega$-regular and positional over state-labelled arenas, then it is $\omega$-star-free.
\end{theorem}




    

\section{Varieties}
\label{sec:varieties}

In the following section we will use `positional' to refer to properties positional over edge-labelled arenas. {\sf LTL} allows us to compositionally build formulae to specify properties, using temporal operators and Boolean operations. However, these operations do not respect positionality; for example: $Fp \wedge Fq$ is not positional, even though both $Fp$ and $Fq$ are. We might wonder whether it is possible to construct some similar specification language closed under Boolean operations which is restricted to expressing only positional languages, thereby providing a compositional syntax for specifying positional properties. To this end, we will investigate varieties of positional languages - classes of positional languages closed under certain combinations of Boolean operations. It is known that varieties of languages correspond to varieties of $\omega$-semigroups, so the (non)existence of varieties of $\omega$-semigroups that recognise only positional languages will determine the (non)existence of Boolean-closed classes of positional languages. An overview of recognition of languages by $\omega$-semigroups can be found in \cite{200475}, and an overview of the correspondence between language varieties and varieties can be found in \cite{2004265}. For algebraic convenience, we will be deal with sets of mixed finite and infinite words, i.e. subsets of $\Sigma^\infty$. However, our claims will only focus on the infinite part of the languages, so for any definition relating to $\omega$-languages (prefix-independence, positionality, etc.), we take $L\subseteq\Sigma^\infty$ to meet these definitions when $L\cap \Sigma^\omega$ meets these definitions.


Firstly, we find that prefix-independent languages are closed under Boolean operations, taking residuals, and pre-images of free $\omega$-semigroup morphisms, so form an $\infty$-variety. Furthermore, each corresponding $\omega$-semigroup $(S_+,S_\omega)$ has the property that the mixed product mapping $S_+\times S_\omega\rightarrow S_\omega$ is right projection, so $uv \mapsto v$ for all $u\in S_+,v\in S_\omega$. This gives us the following:

\begin{restatable}{proposition}{preindvariety}\label{prop:prefindvariety}
    Prefix-independent languages form an $\infty$-variety corresponding to the identity $uv = v$ for $u\in \Sigma^*,v\in\Sigma^\omega$.
\end{restatable}

However, if we attempt to find within this a variety consisting of \textit{positional} prefix-independent languages, we can only find trivial examples:

\begin{restatable}{proposition}{preindtrivvariety}\label{prop:prefindimpossible}
    If $V$ is an $\infty$-variety consisting entirely of prefix-independent positional languages, then the only $\omega$-languages it contains are $\emptyset$ and $\Sigma^\omega$.
\end{restatable}

\begin{remark}
    Since closure under complement means these languages are bipositional, and bipositional prefix-independent languages encode parity conditions \cite{COLCOMBET2006190}, this can equally be phrased as the nonexistence of a non-trivial class of parity conditions closed under Boolean operations.
\end{remark}

This places some restrictions on varieties of positional languages in general. Suppose we have a variety $\mathcal{V}$ entirely consisting of positional languages. If this contained a prefix-independent language $L$, the variety generated by $L$ must be prefix-independent (by \Cref{prop:prefindvariety}) and must be a subvariety of $\mathcal{V}$. This would give us a variety of prefix-independent languages, which we know to be impossible by the above. Therefore:

\begin{restatable}{corollary}{nogo}
    \label{coro:no_variety_pos_preind}
    An $\infty$-variety containing only positional languages must not contain a (non-trivial) prefix-independent language.
\end{restatable}

However if we weaken our closure requirements, we can find a positive $\infty$-variety (an $\infty$-variety minus the requirement to be closed under complement) of prefix-independent positional languages:

\begin{restatable}{proposition}{fgvariety}
    The class of languages containing, for each alphabet $\Sigma$, the languages of the form $FG(A_1) \cup \ldots \cup FG(A_n)$ where $A_i\subseteq \Sigma$ form a positive $\infty$-variety.
\end{restatable}

This suggests we may be able to find more specification languages for positional properties if we drop the need for closure under all Boolean operations.

\section{Applications to Temporal Logic}\label{sec:temporallogic}

In this section we apply the above characterisation of positional properties to identify fragments of the temporal logic {\sf ATL}* in which we can assume agents use positional strategies. We first introduce some notions specific to this section.

${\sf ATL}^*$ extends {\sf LTL} by allowing {\it path formulae}, defined in {\sf LTL}, to be bound to modalities $\langle C\rangle$ which determine (in the semantics) a set of subsets of traces from a state, according to what the coalition of agents $C$ can force.
A formula $\langle C\rangle\varphi$, where $\varphi$ is an {\sf LTL} formula, will be satisfied at a state $x$ when {\it there exists} a subset of traces that $C$ can force from $x$, such that {\it every} trace within satisfies $\varphi$.
An implicit parameter of this logic is the set of agents ${\sf Ag}$.

In this paper we often consider fragments of ${\sf ATL}^*$ generated by a fragment of ${\sf LTL}$ which are allowed as the path formulae.
To this end, we define the syntax of the logic as parameterised over a subset $\mathcal{A}$ of {\sf LTL} formula which possibly contain free variables.
The grammar for ${\sf ATL}+\mathcal{A}$ is separated into {\it state formula} $\psi$  and path formula $\varphi$.
For any $p\in Prop,C\subseteq Ag,\alpha(x_1,\ldots,x_n)\in\mathcal A$:
\begin{align*}
    \psi ::= p \mid \psi \vee \psi \mid \neg\psi \mid \langle C\rangle \varphi && \varphi ::=  X\psi \mid \psi U \psi \mid \psi R \psi \mid \alpha[\psi_1\ldots \psi_n]
\end{align*}
The semantics of the logic is taken over {\it concurrent game structures}, we refer to \cite[Chapter II, 9]{Demri_Goranko_Lange_2016} for details. When $\mathcal{A}=\emptyset$ we get plain {\sf ATL};
when $\mathcal{A}$ consists of all {\sf LTL} formulae possibly with free variables we obtain the full logic ${\sf ATL}^*$;
when we take $\mathcal{A}$ to contain $FG$ and $GF$ (with free variables in the obvious places) we get {\sf EATL};
when we take $\mathcal{A}$ to contain Boolean combinations of $X$, $U$, and $R$ we get ${\sf ATL}^+$;
and when we take $\mathcal{A}$ to take Boolean combinations of $X$, $U$, $R$, $GF$, and $FG$ we get ${\sf EATL}^+$.
When we impose that there is precisely one agent, we get {\sf CTL}, ${\sf CTL}^*$, ${\sf ECTL}$, ${\sf CTL}^+$, and ${\sf ECTL}^+$, using the values of $\mathcal{A}$ in the previous sentence.
In \Cref{sec:temporallogic} we will consider fragments ${\sf ATL}+\mathcal{A}$ where we take $\mathcal{A}$ to contain single formula with free variables, we shall abuse notation and omit the set notation and the free variables; for example, we will say ${\sf ATL}+GU$ to mean $ATL$ with the path formula $G(\psi U \psi)$ added to the syntax.

We can give a {\it positional semantics} for ${\sf ATL}+\mathcal{A}$ by taking the subsets of traces from a state $x$ as the subsets which arise from {\it positional strategies} for $C$.
Furthermore, we can consider a {\it bipositional semantics}, when each trace in each subset (which arises from a positional $C$-strategy) arises from a positional strategy for ${\sf Ag}\setminus C$.
It is a well-known fact that the positional semantics of ${\sf ATL}$ agrees with the regular semantics which use memoryful strategies
(in fact, the bipositional semantics of {\sf ATL} also agrees with the regular semantics). We call some instance of ${\sf ATL}+\mathcal{A}$ a \textit{(bi)positional fragment} if across all expressible formulae and models the regular semantics agrees with the (bi)positional semantics.

\subsection{Positional Fragments}

In the domain of temporal logic, \Cref{coro:no_variety_pos_preind} tell us that we must restrict the path formula in {\sf ATL}* drastically
to obtain a {\it positional semantics} - one where only positional strategies are considered - which agrees with the memoryful semantics.
Despite this, we examine two such restrictions. We consider two fragments of {\sf ATL}$^*$
where we add $GF\psi_1\land FG\psi_2$,
and $G(\psi_1 U\psi_2)$, to the path formula of {\sf ATL} respectively, as described above.
We see these as natural choices of path formula,
because in one-agent {\sf ATL} (i.e. {\sf CTL}), these give us the ability to express Rabin objectives (which subsume parity,
and maintain positionality wrt the first player),
and Rabin objectives {\it with} a safety objective, respectively.
These claims are substantiated shortly in \Cref{prop:rabin_safety}.

As mentioned in the introduction, we achieve efficient model checking for restrictions of {\sf ATL}* with
positional path formula, by reducing the problem to {\sf LTL} model checking (which is in PSPACE).
Given a formula $\varphi$ in the positional fragment of {\sf LTL},
to determine whether $\varphi$ holds at a state $x$, it suffices to (non-deterministically)
enumerate positional Player 1 strategies in the game, and do {\it universal} model checking on the {\sf LTL} formula $\varphi$ on the resulting transition system at $x$. We can say more when the formula $\varphi$ is bipositional. First note that {\it existential} {\sf LTL} model checking of a positional path formulae $\varphi$ is in NP, as we can guess a positional strategy to generate
an ultimately periodic trace $pl^\omega$ where $|pl|$ is at most the number of states, which can be checked in time $O(|pl|\cdot|\varphi|)$\cite{markey:hal-01194626}.
So, if $\varphi$ is bipositional, then $\neg\varphi$ is positional and checking $A \varphi\equiv \neg E\neg\varphi$ is in co-NP. Therefore, {\sf ATL}* model-checking of bipositional properties is in $\Sigma^p_2$ (i.e. in NP given a co-NP oracle \cite{Arora_Barak_2009}), by guessing a positional strategy for P1 and using the oracle to verify $A \varphi$.
These facts yield the following:

\begin{restatable}{proposition}{synthesiscomplexity}
    \label{prop:synthesiscomplexity}
    The model checking problem for any positional fragment of {\sf ATL}$^*$ is in PSPACE.
    The model checking problem for any bipositional fragment of {\sf ATL}$^*$ is in $\Sigma^p_2$.
\end{restatable}


We now verify that adding various path formula of interest to {\sf ATL} maintains positionality of the semantics.

\begin{restatable}{proposition}{positionalformulas}
    \label{prop:pos_formulae}
    If we add any of the following composite {\sf LTL} modalities
    to the path formula of {\sf ATL}, the semantics remains positional:
    $GF\psi_1\land FG\psi_2$,
    $G(\psi_1 U\psi_2)$,
    $GF\psi_1\land FG\psi_2\land G\psi_3$, and $G(\psi\vee X\psi)$.
\end{restatable}
\begin{proof}[Proof sketch]
    The first formula generates a prefix-independent set of traces, so it suffices to verify the conditions of \Cref{prop:pre_ind_conditions}. 
    For the remaining three, the claim follows from verifying the conditions in \Cref{cor:omega-pos}, as every property positional on edge-labelled arenas is positional on state-labelled arenas.
\end{proof}

When adding in two of the path formula considered in \Cref{prop:pos_formulae} to {\sf ATL},
it turns out we obtain logics of equal expressive power.

\begin{restatable}{proposition}{mutualexpress}
    \label{prop:mutualexpress}
    ${\sf ATL}+GF\wedge FG\wedge G$ and ${\sf ATL} + GU$ are mutually expressive.
\end{restatable}
\begin{proof}[Proof sketch]
    The following equivalences can be routinely verified:
    $\langle C\rangle(GF\psi\land FG\varphi\land G\phi)\equiv \langle C\rangle(\phi U (\langle C\rangle(G((\varphi\land\phi)
    U(\psi\land\varphi\land\phi)))$ and
    $\langle C\rangle(G\psi U \varphi) \equiv \langle C\rangle(GF\varphi \wedge FG\top \wedge G(\psi \vee \varphi))$.
\end{proof}


\begin{remark}
    \label{remark:gr1}
 GR(1) synthesis \cite{BLOEM2012911} restricts LTL specifications to an implication from a conjunction of B\"uchi conditions to a conjunction of B\"uchi conditions, to obtain a polynomial-time synthesis algorithm. 
 Our positional fragments
 are incomparable to this approach,
 as ${\sf FG}p\land {\sf GF}q$ is not directly expressible in GR(1),
 and we cannot express conjunctions of B\"uchi conditions as they are not positional
 (and are inexpressible in our fragments, see \ref{prop:ectlplusgu}).
 We leave a detailed comparison to future work.
\end{remark}

Over models with a single player, we now show that adding $GF\varphi\land FG\psi$ ($G\varphi U\psi$) to path formula lets us express Rabin conditions (mixed Rabin and safety conditions).
Recall a Rabin condition is a finite collection of pairs $(U_i,V_i)$, where $U_i$ must be visited infinitely often, and $V_i$ must be visited finitely often.
We say we can {\it encode} a Rabin condition in a fragment of {\sf CTL}* when, given $U_i$ and
$V_i$ are denoted by propositional variables $u_i$ and $v_i$,
we can write a formula which holds precisely when the Rabin condition is satisfied.
In a similar fashion, we can speak of safety conditions, where some set $W$ must never be visited, being encoded in a fragment of {\sf CTL}*.
The following proposition  relies on disjunctions distributing over $E$ in {\sf CTL}*, thus may not apply to the corresponding fragments of {\sf ATL}*.

\begin{proposition}
    \label{prop:rabin_safety}
    The fragment ${\sf CTL} + GF\wedge FG$ can encode a Rabin condition.
    The fragment ${\sf CTL}+GU$ can encode a Rabin condition
    combined with a safety condition.
\end{proposition}
\begin{proof} We show the second claim.
    Given a Rabin condition $(u_i,v_i)$ indexed by a finite set $I$ and a safety condition $w$
    (both encoded as atomic propositions),
    we can write a formula equivalent to
    $E(\bigvee_{i\in I}(GF u_i\land FG v_i)\land Gw)$,
    using distributivity of $E$ over $\lor$, and \Cref{prop:mutualexpress}.
\end{proof}

Note that parity conditions are subsumed by Rabin conditions, so can be expressed by both fragments in \Cref{prop:rabin_safety}.

\subsection{Inexpressivity results}
\label{subsec:inexpress}

In the remainder of this section, we make various inexpressivity claims for our different positional fragments,
along with {\sf EATL} (where $GF$ and $FG$ path formula are permitted), and {\sf EAT$\text{L}^+$} (where we further allow Boolean combinations of path formula). We collect the claims in the following theorem,
the proofs of which will follow.
\begin{theorem}\label{theorem:ATL_fragments} Let $>$ denote the strictly more expressive than relation,
    we have the following:
    \begin{center}
        $
        {\sf ATL}\overset{\cite{10.1007/3-540-47813-2_20}}{=}
        {\sf ATL}^+\,\overset{(\ref{coro:atl_eatl})}<\,
        {\sf EATL}\,\overset{(\ref{prop:eatl_gffg})}<\,
        {\sf ATL}+GF\land FG\,\overset{(\ref{prop:gffgtogu})}<\,
        {\sf ATL}+GU\,\overset{(\ref{prop:ectlplusgu})}<\,\,
        {\sf EATL}^+\,\overset{(\ref{coro:gpxp_eatl})}<\,
        {\sf ATL}^*
        $
    \end{center}
    Furthermore, ${\sf EATL}^+$ cannot express all state-labelled positional properties.
\end{theorem}

Note that {\sf EATL} is positional whereas {\sf EATL$^+$} is not, so as far as we are aware {\sf ATL+$GU$} is the largest known fragment of {\sf ATL}* containing only positional path formulae.

Recall that the standard procedure for showing inexpressibility results in temporal logic
(see e.g. \cite[Lemma 10.3.2]{Demri_Goranko_Lange_2016}) is to show that a formula in one logic picks out a class of models,
which cannot be expressed by the other.
It turns out, that to compare our fragments, it is often sufficient to perform this procedure on just the one-agent restrictions of our logics:

\begin{restatable}{lemma}{express}
    \label{lem:express_transfer}
    Let $A$ be a set of (possibly composite) {\sf LTL} modalities, and $\varphi$ be an {\sf LTL} formula.
    Suppose we have two families of rooted transition systems $\{\mathcal{M}_n\}_{n\in\mathbb{N}}$
    and $\{\mathcal{N}_n\}_{n\in\mathbb{N}}$,
    for which $\mathcal{M}_n\vDash\exists\varphi$ and $\mathcal{N}_n\nvDash\exists\varphi$ for all $n$,
    and where $\mathcal{M}_n$ and $\mathcal{N}_n$ cannot be distinguished by any $CTL+A$ formula of modal depth less than $n$.
    This is enough to show that $\langle i\rangle\varphi$ is not expressible in ${\sf ATL}+\mathcal{A}$, for some agent $i\in{\sf Ag}$.
\end{restatable}




In \cite[Theorem 4.4]{10.1145/567067.567081}, the authors show that $CTL^+<ECTL$
by proving $E(GFp)$ is not expressible in $CTL^+$ using the assumption of \Cref{lem:express_transfer}, thus we have as a corollary:

\begin{corollary}
    \label{coro:atl_eatl}
    $\langle i\rangle GFp$ is not expressible in ${\sf ATL}^+$.
\end{corollary}

To separate ${\sf EATL}$ and ${\sf ATL}+GF\land FG$, we can show the assumptions of \Cref{lem:express_transfer} are satisfied by adapting the proof of \cite[Lemma 2.8.4]{larou1994}, where it is shown that $E(GFp\land Gq)$ distinguishes two classes of models for which no ECTL formula can.
It turns out that $E(GFp\land FGq)$ also distinguishes the same two classes of models.

\begin{restatable}{proposition}{eatltogffg}
    \label{prop:eatl_gffg}
    $\langle i\rangle(GFp\wedge FGq)$ is not expressible in ${\sf EATL}$.
\end{restatable}

It easily follows from \Cref{prop:mutualexpress}, that ${\sf ATL}+GU$ is as at least as
expressive as ${\sf ATL}+GF\land FG$. To show it is strictly more, we can construct two classes of rooted transitions systems which satisfy \Cref{lem:express_transfer} by hand.

\begin{restatable}{proposition}{gffgtogu}
    \label{prop:gffgtogu}
    $\langle i\rangle(GpUq)$ is not expressible in ${\sf ATL} + GF\wedge FG$.
\end{restatable}

We can easily see ${\sf EATL}^+$ is at least as expressive as ${\sf ATL}+GU$, by using \Cref{prop:mutualexpress} and the fact that $GF\varphi\land FG\psi\land G\psi$ is a Boolean combination of ${\sf EATL}$ path formulas.
To get that ${\sf EATL}^+$ is strictly more expressive, we can show the models given in the proof of $ECTL<ECTL^+$ in \cite[Theorem 4.3]{10.1145/567067.567081} cannot be distinguished by ${\sf ATL}+GU$ formulas, and apply \Cref{lem:express_transfer}.

\begin{restatable}{proposition}{ectlplusgu}
    \label{prop:ectlplusgu}
    $\langle i\rangle(GFp\wedge GFq)$ is not expressible in ${\sf ATL} + GU$.
\end{restatable}


The formula $EG(p\lor Xp)$ was shown not to be expressible in $ECTL^+$ in \cite[Theorem 4.2]{10.1145/567067.567081}.
Again, we can use two classes of models to apply \Cref{lem:express_transfer},
thus we obtain:
\begin{corollary}
    \label{coro:gpxp_eatl}
    $\langle i\rangle G(p\lor Xp)$ is not expressible in ${\sf EATL}^+$
\end{corollary}
This corollary tells us two things: firstly that ${\sf ATL}^*$ is strictly more expressive than ${\sf EATL}^+$,
and secondly that not all positional properties are expressible in ${\sf EATL}^+$ (we saw in \Cref{prop:pos_formulae} that $G(p\lor Xp)$ is a positional property).

    


\section{Conclusion}

 We have shown that positional $\omega$-regular languages are $\omega$-star-free, and provided a language-theoretic characterisation of these languages to complement the automata-theoretic characterisation of \cite{theoretics:14945}. We have applied these results to problems in temporal logic. Firstly, a negative result showing we cannot find a property specification language closed under Boolean operations which expresses only positional and prefix-independent properties. Secondly, we have presented various fragments of {\sf ATL}$^*$
 with positional path formulae (and so more tractable model-checking), and shown where they fall expressively in the hierarchy of {\sf ATL}$^*$ fragments.
 
 We present directions for future work. Further work is required to refine exactly where positional $\omega$-regular languages are situated within the various hierarchies of sub-$\omega$-regular languages. We also do not know whether there are $\omega$-regular 1P non-uniform positional languages which are not $\omega$-star-free. On another note, totally-ordered residuals are a key requirement for positionality, but they are not entirely understood. Some work has been done for the finite case in \cite{refId0}, where they are termed `Ferrers languages'. Due to \Cref{cor:edgeunion}, it would be good to know exactly when the union of two languages has totally ordered residuals. It would also be interesting to know whether there are any non-trivial varieties consisting entirely of positional languages. Additionally, in the domain of finite languages, there are variety theorems for several restricted combinations of Boolean operations, e.g. disjunctive varieties \cite{Polak2004}, basic varieties \cite{10.1007/978-3-030-68195-1_1}. If these can be extended to $\omega$-languages, then we can find classes of positional languages which are instances of these weaker types of variety, which may provide canonical algebraic identities to classify these languages. Finally, since {\sf ATL}$^+$ is expressible in {\sf ATL} \cite{10.1007/3-540-47813-2_20}, there may be an analogue of \Cref{prop:rabin_safety} which holds for {\sf ATL}$^*$ through a similar construction.

\bibliography{lipics-v2021-sample-article}

\newpage
\appendix
\section{Omitted Proofs}
\label{appendix:proofs}

\starfree*
\begin{proof}
    $(\implies)$ Suppose $L$ is $\omega$-star-free. This means it is recognised by a finite aperiodic monoid. That is, a monoid $M$ where for each $x\in M$, it is the case that for some $n$, $x^n=x^{n+m}$ for all $m$. Call the monoid $M$ and the recognising homomorphism $h$. First, we show counter-freeness. Take some $uvw\in \Sigma^\omega$. We know that $M$ is aperiodic, so $h(v)^n = h(v)^{n+m}$ for some $n$ and all $m$. Since $h$ is a monoid homomorphism, $h(v^n) = h(v^{n+m})$. Therefore, $uv^nw\sim _h uv^{n+m}w$ so $uv^nw\in L$ iff $uv^{n+m}w\in L$. We now have an $n$ for each $u,v,w$ at which membership of $L$ no longer changes, but to show the language is counter free we need to show there is single $n$ that works across every word. Since $M$ is finite, we can take the least $n$ such that $x^n=x^{n+1}$ for all $x\in M$. Taking this $n$, we know that there exists an $n$ such that $uv^nw\in L$ iff $uv^{n+m}w\in L$ as required. 

    The argument is similar to show $L$ is $\omega$-counter-free. We know that given some $u,v,x,y\in \Sigma^*$, taking the $n$ for which $h(x^n) = h(x^{n+m})$, that $u(vx^ny)^\omega \sim_h u(vx^{n+m}y)^\omega$, and therefore $u(vx^ny)^\omega\in L$ iff $u(vx^{n+m}y)^\omega\in L$. Again, we take the least $n$ such that $x^n = x^{n+1}$ for all $x\in M$ and we have that there exists an $n$ such that for all $u,v,x,y\in \Sigma^*$ it is the case that $u(vx^ny)^\omega\in L$ iff $u(vx^{n+m}y)^\omega\in L$ as required.

    $(\impliedby)$ Given an $\omega$-regular $L$, let us suppose it also is counter-free and $\omega$-counter-free. Since it is $\omega$-regular, then by \cite{ARNOLD1985333} the following syntactic congruence is finite and saturates $L$:\\~

    $w\approx w'$ iff for all $u,v_1,v_2\in \Sigma^*$, $v_1wv_2u^\omega \in L$ iff $v_1w'v_2u^\omega \in L$ and $u(v_1wv_2)^\omega\in L$ iff $u(v_1 w' v_2)^\omega$\\~

    We take the finite monoid $\Sigma^*/\approx$ where multiplication is defined as $[s] \cdot [t] = [s\cdot t]$. It can be seen that the function $[-]:\Sigma^*\rightarrow \Sigma^*/\approx$ that maps $x\mapsto [x]$ is a monoid homomorphism that recognises $L$. We will show that $\Sigma^*/\approx$ is aperiodic, giving us a finite aperiodic monoid that recognises $L$.

    Take some arbitrary $[x] \in \Sigma^*/\approx$. We know since $L$ is CF and $\omega$-CF there exists an $n_1$ such that for any $u,v\in \Sigma^*$,$w \in \Sigma^\omega$, it is the case that $uv^{n_1}w\in L$ iff $uv^{n_1+m}w\in L$ for all $m$. We also know there exists an $n_2$ such that for any $u,v,x,y\in \Sigma^*$, $u(vx^{n_2}y)^\omega\in L$ iff $u(vx^{n_2+m}y)^\omega\in L$ for all $m$. Take the largest of $n_1,n_2$ and call this $n$. We will show $[x]^n = [x]^{n+m}$ for all $m$, for which it is sufficient to show $[x]^n = [x]^{n+1}$. We know $[x]^n = [x^n]$ and $[x]^{n+1}=[x^{n+1}]$, so we just need to show that $x^n\approx x^{n+1}$.

    Take some arbitrary $u,v_1,v_2\in \Sigma^*$. Suppose $v_1 x^n v_2 u^\omega\in L$. By CF, this is the case iff $v_1x^{n+1}v_2u^\omega\in L$. Suppose $u(v_1 x^n v_2)^\omega \in L$. By $\omega$-CF, this is the case iff $u(v_1 x^{n+1} v_2)^\omega\in L$. Therefore, $x^n\approx x^{n+1}$ as required. So, $L$ is recognised by a finite aperiodic monoid and is therefore $\omega$-star-free.
\end{proof}

\residualcounterfree*
\begin{proof}
    Suppose we have a language $L\subseteq \Sigma^\omega$ which has a finite, totally ordered set of residuals $R(L)$. We will first show that for any choice of $u,v\in \Sigma^*$, $w\in \Sigma^\omega$ the word $uv^iw$ can only `switch' membership of $L$ at most once as $i$ increases from 0. In other words, if $uv^0w=uw\in L$ and there is some $uv^iw\notin L$, then $uv^jw\in L$ for $j< i$ and $uv^kw \notin L$ for $k \geq i$, and vice versa. To show this, first suppose this was not the case. We begin with the case where $uw\in L$. We would have some $i,j\in \mathbb{N}$ with $0<i<j$ such that $uv^iw\notin L$ and $uv^jw\in L$. If we take $i$ and $j$ to be the least such numbers where the language switches membership, we know that for $0\leq p < i$, $uv^pw\in L$, and for $i \leq q < j$, $uv^qw\notin L$. Let us look at the residuals for the words $uv^{(i-1)}$ and $uv^{(j-1)}$. We know that $w\in R_L(uv^{(i-1)})$ and $w\notin R_L(uv^{(j-1)})$, as $uv^{(i-1)}w\in L$ but $uv^{(j-1)}w\notin L$. So, $R_L(uv^{(i-1)})\not\subseteq R_L(uv^{(j-1)})$. We know that $vw\notin R_L(uv^{(i-1)})$ and $vw\in R_L(uv^{(j-1)})$, as $uv^{i}w\notin L$ but $uv^{j}w\in L$. Therefore, $R_L(uv^{(j-1)})\not\subseteq R_L(uv^{(i-1)})$. So, the set of residuals cannot be totally ordered, contradicting our assumption. The argument for the case where $uw \notin L$ is analogous. Therefore, it must be that for any choice of $u,v,w$, the word $uv^iw$ switches membership at most once as $i$ increases from 0.

    Given this, we will assume that $L$ is not counter-free. Therefore, for all $n$, there exists some $u,v,w,m\in \mathbb{N}$ such that $uv^nw\in L$ iff $uv^{n+m} w\notin L$. So for any $n$, we can find a word with substring $v$ where the switch in membership happens at some point after $v$ is iterated $n$ times. Let us call the point at which the membership switches $m_n$ for each $n$. Then we have the following, where $m_n > n$:

    $$n=0, u_0w_0\in L, u_0v_{0}^{m_0}w_0\notin L$$
    $$n=1, u_1v_{1}w_1\notin L, u_1v_{1}^{m_1}w_1\in L$$
    $$n=2, u_2v_{2}^2w_2\notin L, u_2v_{2}^{m_2}w_2\in L$$
    $$...$$

    Although this example is just illustrative; the switch in membership may happen in the other direction for any $n$. 
    
    Note that for a given $n$, we generate at least $n$ different residuals. To see this, suppose for some $n=i$, $uv^i w \in L$ and $u v^{m_i} w\notin L$. We know membership of the language can only switch once, so $uv^jw \in L$ for $j < m_i$ and $uv^jw \notin L$ for $j \geq m_i$. Then we can look at the residuals for each $uv^k$ where $k$ ranges from $0\leq k \leq m_i$. We can see that $uv^{k+l}w\in L$ iff $k + l < m_i$. This means the residual $R_L(uv^k)$ must be distinct from any residual $R_L(uv^{k-j})$ for $1\leq j \leq k$, as the latter contains $v^{m_i - k}w$ whereas the former does not, as this would take it out of the language. So, each of the residuals $R_L(uv^0)$ up to $R_L(uv^k)$ are distinct sets. We can argue symmetrically for the case when $uv^iw\notin L$ and $uv^{m_i}w\in L$.

    Since the set of residuals $R(L)$ must at least contain each of these residuals for each $n$, it must be infinite. If it were finite, say of size $k$, then it could not contain all of the distinct residuals for the case of $n=k+1$. This contradicts our assumption that $R(L)$ is finite, so $L$ must in fact be counter-free.
\end{proof}

\regposcounterfree*
\begin{proof}
    Take a language $L\subseteq \Sigma^\omega$ which is $\omega$-regular and positional. Take some arbitrary $u,v,w,x\in \Sigma^*$. We will first show that a word $u(vx^nw)^\omega$ cannot be out of $L$ at some point $n=i$, in the language at some point $n=j$ where $j>i$, and again out of the language at some point $n=k$ where $k>j$. This means that for every word $u(vx^nw)^\omega$, there is a finite $n$ at which increasing $n$ no longer changes membership of the language. First, for some $u,v,w,x\in \Sigma^*$ take $i$ to be some value where $u(vx^iw)^\omega\notin L$, take $j$ to be the first value such that $j>i$ and $u(vx^jw)^\omega\in L$, and take $k$ to be the first value such that $k>j$ and $u(vx^kw)^\omega\notin L$. Let us construct the following 2-player game with $L$ as a winning condition, where P1 controls the square nodes:
\begin{center}
    \begin{tikzpicture}
        \node[rectangle, draw] (n1) {};
        \node[rectangle,draw,right=2cm of n1] (n2) {};
        \node[circle, draw, right=2cm of n2] (n3) {};
        \node[rectangle,draw,right=2cm of n3] (n4) {};
        \node[rectangle,draw,right=2cm of n4] (n5) {}; 
        \node[right=0.6cm of n4] (dots) {$\ldots$}; 
        \node[right=0.6cm of n3] (dots1) {$\ldots$}; 

        \draw (n1) -> (n2) node [midway, above, fill=white] {$u$};
        \draw (n2) -> (n3) node [midway, above, fill=white] {$v$};

        \draw[->] (n3) to[bend left]
          node[midway, above, fill=white] {$x^{j-1}$}
          (n4);
        
        \draw[->] (n3) to[bend right]
          node[midway, above] {$x^{k-1}$}
          (n4);
        
        \draw[->] (n4) to[bend left]
          node[midway, above, fill=white] {$\epsilon$}
          (n5);

        \draw[->] (n4) to[bend right]
          node[midway, above] {$x^{k-j}$}
          (n5);

        \draw[->] (n4) to
          node[midway, above] {$x$}
          (n5);
        
        \draw[->] (n5.south) to[bend left]
          node[midway, above, fill=white] {$w$}
          (n2.south);
    \end{tikzpicture}
\end{center}
    Note that P1 has a memoryful strategy which guarantees all traces are in $L$, by always aiming for a total of $x^{k-1}$. However, no choice of positional strategy can guarantee that P1 wins, as P2 can always force either $x^{j-1}$. or $x^k$ However, this is contradictory as we know by assumption that $L$ is positional - so it must be the case there is always some value $n$ after which $u(vx^nw)^\omega$ does not change membership of $L$.

    Since $L$ is $\omega$-regular, we can use the syntactic congruence $\approx$ of \cite{ARNOLD1985333} to obtain a finite monoid that recognises $L$. Suppose this monoid is of size $k$. For each element $[x]\in M$, it must be that there is $[x]^n = [x]^m$ for some $n\leq k$, $n\leq m\leq k$. Furthermore, $[x]^n = [x]^{n+i(m-n)}$ for all $i$. So, iterating $[x]$ eventually gives us a cycle of period $m-n$. Since $[x]^n = [x]^m = [x^n] = [x^m]$, this also means that $x^n\approx x^{n+i(m-n)}$, so for any $u,v,w\in \Sigma^*$, $u(vx^nw)^\omega\in L$ iff $u(vx^{n+i(m-n)}w)^\omega\in L$.

    We will show that for any word of the form $u(vx^iw)^\omega$, the $i$ after which it no longer switches from being in $L$ to outside of $L$ or vice versa is bounded by $k=|\Sigma^*/\approx| + 1$. Let $n,m$ be the points at which $[x]^n=[x]^m$ as detailed above. Suppose the $i$ at which the word switches occurs on $i=(n + j(m-n))$ for some $j$. Since, by the above, $u(vx^nw)^\omega\in L$ iff $u(vx^{n+i(m-n)}w)^\omega\in L$, the switch would have to happen at $u(vx^nw)^\omega$, and $n < k$. Suppose the $i$ at which the final switch occurs at some $n+ j(m-n)<i<n+ (j+1)(m-n)$ for some $j>0$. Either $u(vw^{(n+j(m-n))}x)^\omega$ (and therefore also $u(vw^{(n+(j+1)(m-n))}x)^\omega$) is in $L$ and $u(vw^{i}x)^\omega$ is not in $L$ or vice versa. However, as $i$ is the value of the final switch, $u(vw^{(n+(j+1)(m-n))}x)^\omega$ must be in $L$ iff $u(vw^{i}x)^\omega$ is in $L$, which contradicts the previous statement. So, it must be that if the word switches membership it occurs at $n$ or earlier. Therefore, the switch must occur before $k$, so $k$ is such that for any $u,v,w,x\in \Sigma^*$, $u(vw^kx)^\omega \in L$ iff $u(vw^{k+l}x)^\omega\in L$ for all $l$, as required for $\omega$-CF.
\end{proof}

\finiteunion*
\begin{proof}
    Given a finite set $\mathbf{L}$ of such languages, take the union $\cup \mathbf{L}$. Take any trace $uvw^\omega\in \cup\mathbf{L}$; this trace must also be in some $\omega$-regular non-uniform 1P positional language $L\in \mathbf{L}$. Therefore either $uv^\omega \in L$ or $uwx^\omega\in L$, so one of these is in $\cup\mathbf{L}$. Similarly for any trace $u(vwx)^\omega\in \cup\mathbf{L}$. Therefore $\cup \mathbf{L}$ meets both conditions for positionality, and as a finite union of $\omega$-regular languages must also be $\omega$-regular.
\end{proof}

\nonuniformpositional*
\begin{proof}
    ($\impliedby$) 
    Suppose we have a winning strategy from a given initial state $s$ over a 1P arena $(V,\delta)$ with winning condition $L$. Since $L$ is $\omega$-regular, we have a finite-memory winning strategy with memory states $M$. Let us consider the play beginning at $s$ generated by this strategy. There are at most $|V|\times |M|$ possible configurations of states and memory states, and clearly the play will begin looping once the same configuration is seen twice, which must happen within $|V|\times |M| + 1$ steps. Therefore this strategy will generate an ultimately periodic play $pl^\omega\subseteq\delta^\omega$ with prefix $p\subseteq\delta^*$ and infinitely looping section $l\subseteq\delta^+$. Suppose in the course of the play, we see the same state twice but select a different edge each time, meaning our strategy is not positional. i.e. the play contains for some $v\in V$, two edges $e_1,e_2\in \delta$ such that $src(e_1)=src(e_2)=v$ but $e_1\neq e_2$. We can use conditions (1) and (2) to modify our play to remove this violation of positionality in such a way that we still generate a winning ultimately periodic play. We will count the number of violations in a finite sequence of edges $p\in\delta^+$ as $vio(p)= \Sigma_{v\in V}count(p,v)$ where $count(p,v) = 0$ if for all edges $e_1,e_2$ in $p$ such that $src(e_1)=src(e_2)=v$ then $e_1 = e_2$, and otherwise $count(p,v)$ is the number of indices $i$ in $p$ at which $src(p[i])=v$. It can be seen that once $vio(p) = 0$, then for any $p_1p_2 = p$, the ultimately periodic play $p_1p_2^\omega$ will contain no two edges $e_1,e_2$ where $src(e_1)=src(e_2)$ but $e_1\neq e_2$.
    
    Let us first look at violations where at least one of these violating edges is in the prefix $p$; we will use condition (1) to fix these. Suppose our play generates the trace $uvwx^\omega\in L$. We refer to the segments of the play which generate $u$, $v$, $w$, and $x$ as $p_1,p_2,p_3$ and $p_4$ respectively. Suppose $p_2[0] = e_1$ and $p_3[0] = e_2$ for two violating edges which begin from the same state as described earlier. If $uv^\omega \in L$, we can replace the play $pl^\omega$ with the winning play $p_1p_2^\omega$, which generates the trace $uv^\omega$. Otherwise, $uwx^\omega\in L$, and we replace $pl^\omega$ with $p_1p_3p_4^\omega$, which generates $uwx^\omega$. Note that in the case where one violation is in the prefix (say, at index $i$ of $p$) and one violation is in the loop (say, at index $j$ of $l$) then by following this process we can reach $p[0,i-1]l[j,|l|-1]l^\omega$, for which $vio$ does not necessarily decrease. However, this is equal to $p[0,i-1](l[j,|l|-1]l[0,j-1])^\omega$, and since every edge in $p[0,i-1]l[j,|l|-1]l[0,j-1]$ is contained within $pl$ the same or fewer number of times and with at least one violation removed, $vio$ decreases for this play.

    For violations entirely within $l$, we must alter the strategy to remove this at every instance of $l$ within $pl^\omega$. Therefore, we will require condition (2). Suppose our play generates a trace $u(vwx)^\omega\in L$, where $l = l_1l_2l_3$ st. $l_1$ generates $v$, $l_2$ generates $w$, $l_3$ generates $x$, and $l_2$ and $l_3$ begin on the same state but select different edges. Suppose $uvw^\omega \in L$; then can replace $pl^\omega$ with $pl_1l_2^\omega$. Otherwise, $u(vx)^\omega\in L$ and we can replace $pl^\omega$ with $p(l_1l_3)^\omega$.

    Each application of these conditions replaces our play $pl^\omega$ with a new ultimately periodic play $p'l'^\omega$. Let us denote by \textit{cycle} a sequence of edges $e_1\ldots e_n$ such that $src(e_1) = tgt(e_n)$. It can be seen that $p'l'^\omega$ is a valid play on our arena, since we have either fixed our play to end on a cycle which already existed in our original play, or we have cut out a cycle which existed in our original play. Additionally, it can be seen that $vio(p'l')<vio(pl)$, i.e. $p'l'$ has strictly fewer violations of positionality than $pl$, and moreover $|p'l'|<|pl|$. Since $pl$ is finite, there are finitely many violations of positionality within $pl$, and once $pl$ contains no violations of positionality then $pl^\omega$ contains no violations of positionality. Therefore, we can remove all violations of positionality with a finite number of applications of these conditions, leaving us with a play where a given state is always followed by the same edge. From this it is simple to construct a positional strategy which generates this play from the same initial state.

    ($\implies$) Let us take an $\omega$-regular language $L$ which is non-uniform 1P positional. Suppose for some $uvwx^\omega\in L$, neither $uv^\omega \in L$ nor $uwx^\omega\in L$. Then, $L$ would not be positional in the following game, where P1 controls all nodes:

    
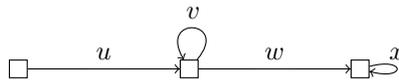
\begin{center}
    \begin{tikzpicture}
        \node[rectangle, draw] (n1) {};
        \node[rectangle,draw,right=2cm of n1] (n2) {};
        \node[rectangle, draw, right=2cm of n2] (n3) {};

        \draw[->] (n1) -> (n2) node [midway, above, fill=white] {$u$};
        \draw[->] (n2) -> (n3) node [midway, above, fill=white] {$w$};

        \draw[->] (n2) to[loop above,looseness=10,in = 120, out = 50]
          node[midway, above] {$v$}
          (n2);
    
        
        
        \draw[->] (n3) to[loop right]
          node[midway, above] {$x$}
          (n3);
    \end{tikzpicture}
    \end{center}
    Suppose for some $u(vwx)^\omega\in L$, neither $uvw^\omega \in L$ nor $u(vx)^\omega\in L$. Then, $L$ would not be positional in the following game where P1 controls all nodes:

    \begin{center}
    \begin{tikzpicture}
        \node[rectangle, draw] (n1) {};
        \node[rectangle,draw,right=2cm of n1] (n2) {};
        \node[rectangle, draw, right=2cm of n2] (n3) {};

        \draw[->] (n1) -> (n2) node [midway, above, fill=white] {$u$};
        \draw[->] (n2) -> (n3) node [midway, above, fill=white] {$v$};

        \draw[->] (n3) to[loop right]
          node[midway, above] {$w$}
          (n3);
    
        
        
        \draw[->] (n3.south) to[bend left]
          node[midway, above, fill=white] {$x$}
          (n2.south);
    \end{tikzpicture}
    \end{center}

    Additionally, in the case where $x$ is empty, the following game provides a counterexample:

        \begin{center}
    \begin{tikzpicture}
        \node[rectangle, draw] (n1) {};
        \node[rectangle,draw,right=2cm of n1] (n3) {};

        \draw[->] (n1) -> (n2) node [midway, above, fill=white] {$u$};

        \draw[->] (n3) to[loop above]
          node[midway, above] {$v$}
          (n3);

        \draw[->] (n3) to[loop below]
          node[midway, below] {$w$}
          (n3);
        
        
    \end{tikzpicture}
    \end{center}
\end{proof}

\poneuniformpos*
\begin{proof}
    For the left-to-right direction, recall that we require totally ordered residuals for positionality.

    For the right-to-left direction, suppose we have a finite-memory strategy over a 1P arena, which is winning from states $W\subseteq V$. For each $w\in W$, the strategy generates an ultimately periodic play $\varsigma_w = p_wl_w^\omega$. We will modify these plays so that whenever a state is reached at multiple different points (possibly across different plays) the same edge is selected each time. 

     We will change the way we count violations of positionality from \Cref{prop:1pnonuniformpos}. This time, given a set of set of plays $P\subseteq \delta^\omega$, the violation count of a state $v\in V$ is $count(P,v) = |\{e\in\delta \mid src(e) = v, e=p[i]\text{ for some $i\in \mathbb{N}$ and some $p\in P$}\}|$. In other words, the violation count of $v$ in $P$ is the number of distinct edges across any play in $P$ which are outgoing from $v$. Once $count(P,v)=1$ for all $v\in V$, we can easily construct a positional strategy which generates all the traces in $P$. We will perform two different operations to lower the violation count of a given set of ultimately periodic plays $P$:

    \begin{enumerate}
        \item Remove violations along a single play $\varsigma_w$ using the method outlined in \ref{prop:1pnonuniformpos}.
        \item Remove violations across multiple plays using totally ordered residuals.
    \end{enumerate}

    We will expand on the second operation. Let us suppose we have first applied operation (1) to remove all violations along each play - so in a given play, for every state $v\in V$, there is a single distinct outgoing edge from $v$ that features in the play. However, the outgoing $v$ edge in one play may not be equal to the outgoing $v$ edge in another play. We will first select some state $v\in V$ which appears across multiple plays but is not followed by the same edge each time. We will split each play $\varsigma_w$ containing $v$ into $\varsigma_w= \varsigma_{w,1}\varsigma_{w,2}$ where $\varsigma_{w,2}$ begins at the first occurrence of $v$ in the play. Let $\lambda_{w,1},\lambda_{w,2}$ be the words generated by $\varsigma_{w,1}$ and $\varsigma_{w,2}$ respectively. Since $L$ has totally ordered residuals, there is some $w\in W$ such that for all other $w'\in W$, $R_L(\lambda_{w',1})\subseteq R_L(\lambda_{w,1})$, which means that for each $w'$, $\lambda_{w',1}\lambda_{w,2}\in L$. Therefore, if we replace each $\varsigma_{w'}$ with $\varsigma_{w',1}\varsigma_{w,2}$ for this maximal $w$, each play is still winning. It can also be seen that $\varsigma_{w,2}$ only contains a single distinct edge outgoing from $v$, namely the edge at $\varsigma_{w,2}[0]$. Since we did this at the first instance of an outgoing $v$ edge across every play, every edge $e$ with $src(e)=v$ across all plays is now equal to $\varsigma_{w,2}[0]$. Additionally, every play is still ultimately periodic. We will call this new set of plays $P$. This means our violation count for $v$ is $count(P,v)=1$. This operation may introduce violations along the plays for some states, in terms of the count used in \Cref{prop:1pnonuniformpos}. However, note that further applications of this operation and operation (1) will preserve the fact that every instance of $v$ is followed by the same outgoing edge, since none of these operations add new edges that did not already exist in some play. So, once we have that $count(P,v)=1$ for some $v$, this cannot be increased by operations (1) or (2). Therefore, once we repeat this process for each $v\in V$ of applying operation (1) along all plays and then operation (2) to set all edges outgoing from $v$ to the same choice, we are left with a set of winning plays from each $w\in W$ where a given state is always followed by the same edge. From this it is simple to construct a positional strategy which generates these plays.

\end{proof}

\wilkealg*
\begin{proof}
    Suppose $F\subseteq S_\omega$ is the set and $\varphi:(A^+,A^\omega)\rightarrow(S^+,S^\omega)$ the map such that $\varphi^{-1}(F) = L$. To check for totally ordered residuals, for each $s\in S_+$ we take the set $s^{-1}L = \{w\in S_\omega | sw \in F\}$. $L$ has totally ordered residuals if $\{s^{-1}L | s\in S_+\}$ is totally ordered by inclusion. This can be checked in time $O(|S_+|\cdot|S_\omega|)$. For each of the other conditions, we need to check $uvwx^\omega$ and $u(vwx)^\omega$ (and also $vwx^\omega$, $u(vw)^\omega$) for each $u,v,w,x\in S_+$ , which can be done in time $O(|S_+|^4)$.
\end{proof}

\preindconds*
\begin{proof}
    Prefix-independent languages have only one distinct residual, so the set of residuals is trivially totally ordered. Suppose $uvwx^\omega\in L$. By prefix-independence, this is the case iff $x^\omega\in L$ iff $uwx^\omega\in L$, meeting the first condition for positionality. Similarly, $u(vw)^\omega\in L$ iff $(vw)^\omega\in L$, $uv^\omega\in L$ iff $v^\omega\in L$ and $uw^\omega\in L$ iff $w^\omega\in L$. So, all conditions for positionality from Cor. \ref{cor:omega-pos} are met when $(uv)^\omega\in L$ implies either $u^\omega\in L$ or $v^\omega\in L$.
\end{proof}

\specialform*
\begin{proof}
    Clearly $L$ is prefix-independent. Let us take a word $(vw)^\omega\in L$. WLOG let us assume $v=v_1sv_2s\ldots sv_n$ and $w=w_1sw_2s\ldots s w_n$. We will show either $w^\omega\in L$ or $v^\omega\in L$. Clearly each $v_i$ or $w_i$ is in $R$, so this amounts to showing either $v_nv_1$ or $w_nw_1$ is in $R$. Suppose both $v_nv_1$ and $w_nw_1$ are not in $R$. However, we know $v_nw_1\in R$ and $w_nv_1\in R$, since $(vw)^\omega\in L$. So, $w_1\in R_R(v_n)$, $w_1\notin R_R(w_n)$ and $v_1\notin R_R(v_n)$, $v_1\in R_R(w_n)$, contradicting our assumption of totally ordered residuals for $R$. So, at least one of $v_nv_1$ or $w_nw_1$ is in $R$, meaning either $v^\omega$ or $w^\omega$ is in $L$ as required. If there are no occurences of $s$ in $v$ or no occurences of $s$ in $w$, we have either $v_nwv_1\in R$ or $w_nvv_1\in R$, from which can use the closure under subwords to guarantee that either $v^\omega$ or $w^\omega\in L$ respectively.
\end{proof}

\ordantidict*
\begin{proof}
    Suppose $L$ has totally ordered residuals. Then, either $R_L(x)\subseteq R_L(u)$ or vice versa - let us suppose the former. Since $v\notin R_L(u)$ then $v\notin R_L(x)$, so $xv\notin L$. Therefore there must be some $d\in D$ such that $d \leq_{SUB} xv$. The argument is symmetric for the other case. 
    
    Suppose for every $xy,uv\in D$, there is some $d\in D$ such that $d\leq xv$ or $d \leq uy$. Let us look at $R_L(u)$ and $R_L(v)$ for some arbitrary $u,v$. Suppose there is some $x\in R_L(u),y\notin R_L(u)$ and $x\notin R_L(v),y\in R_L(v)$. So, there must be some $p,q\in D$ such that $p= p_1p_2$, $q=q_1q_2$, and $p_1\leq_{SUB}u$,$p_2\leq_{SUB}y$,$q_1\leq_{SUB}v$,$q_2\leq_{SUB}x$. We know there is some $d\in D$ such that either $d\leq_{SUB}p_1q_2$ or $d\leq_{SUB}q_1p_2$. However, this means either $d\leq_{SUB}ux$ or $d\leq_{SUB}vy$, which contradicts our assumption they were in the language. So, $L$ must have totally ordered residuals.
\end{proof}

\subwordclosed*
\begin{proof}
    Suppose $L$ is a subword-closed language with anti-dictionary $D=\{w\}$. Let $|w| = l$. Let us take $xy,uv\in \Sigma^*$ such that $xy=uv=w$. Let $|x| = i$ and $|u| = j$. Suppose $i\leq j$. The word $uy$ is equal to $w[1,j] \cdot w[i+1,l]$. Since $j\geq i$, this must contain $w$ as a subword, so $w\leq_{SUB} uy$. The case is symmetric where $j\leq i$. Therefore, $L$ meets the condition for totally-ordered residuals.
\end{proof}

\totordresiduals*
\begin{proof}
    Suppose this is not the case, so $x\in R_L(ua),y\in R_L(va),y\notin R_L(ua),x\notin R_L(va)$ for $u,v,x,y\in \Sigma^*,a\in \Sigma$, s.t. $u=u_0\ldots u_n,v=v_1\ldots v_m$ for $u_i,v_i\in \Sigma$. Then, $L$ is not positional over the following game:

    \begin{center}
    \begin{tikzpicture}
    \node[rectangle,draw] (n2) {$a$};
    \node[circle, draw,below left=0.5cm and 1.5cm of n2] (n1) { $u_0$};
    \node[circle,draw,above left=0.5cm and 1.5 of n2] (n5) {$v_0$};

    \node[above right=1cm of n2] (n3) {$\ldots$};
    \node[below right=1cm of n2] (n4) {$\ldots$};

    \draw[->]
      (n5) to[bend left]
      node[midway, above] {$u_1\ldots u_n$}
      (n2);
    
    \draw[->]
      (n1) to[bend right]
      node[midway, above] {$v_1\ldots v_n$}
      (n2);
    
    \draw[->]
      (n2) to[bend left]
      node[midway, above] {$x$}
      (n3);
    
    \draw[->]
      (n2) to[bend right]
      node[midway, above] {$y$}
      (n4);
\end{tikzpicture}
\end{center}

\end{proof}

\antidict*
\begin{proof}
    Take some $d\in D$ such that $d=d_1d_2$ for $d_1,d_2\in \Sigma^*$. We also have $a=\varepsilon a = a \varepsilon\in D\cup \{a\}$. Both $d_1a$ and $ad_2$ have $a\in D\cup\{a\}$ as a subword, so $D\cup\{a\}$ meets the required condition for totally-ordered residuals from Prop \ref{prop:totally_ordered_anti_dict}.
\end{proof}

\stateposchar*
\begin{proof}[Proof sketch]

    The idea will be to represent a strategy in a 2-player game as a collection of infinite trees, where the infinite paths through the trees give us all plays generated by that strategy. A strategy is then winning if all the infinite paths lie within the objective $L$. The nodes in the tree will correspond to states. A node corresponding to a P1 controlled state will have a single successor, encoding the move chosen by the strategy. A node corresponding to a P2 controlled state will have a successor for every transition from that state. We will show that we can adjust the moves taken at P1 controlled nodes to obtain a tree where for a given P1 state, at every node corresponding to that state we select the same next move. 

    We will first note that condition (1) can be actually be applied to words without a periodic tail, as by the Ramseyan factorization theorem \cite{200475} every word $z\in \Sigma^\omega$ admits a factorization $z=z_1z_2\ldots$ such that for any $u\in\Sigma^*$, $uz_1z_2^\omega\in L$ iff $uz\in L$. So, given a word $uvw\in L$, we take the above factorisation $w_1w_2\ldots$ and replace our word with $uvw_1w_2^\omega$, which must be in $L$. We then apply condition (1) to either obtain $uw_1w_2^\omega\in L$ or $uv^\omega\in L$. If the former, we know this is the case iff $uw\in L$. We will use this to apply condition (1) to infinite paths through the tree of plays which may not necessarily be ultimately periodic.

    Suppose we have a state-labelled 2P game $(V_1,V_2,\delta,\pi)$ with a finite-memory strategy $\sigma$ for P1, which is winning from a set of states $W\subseteq V$. We will assume the strategy is implemented by a memory structure $(M,m_0,\mu)$ with $M$ memory states, and transitions $\mu:M\times V \rightarrow V$ are labelled with states $V$, so the memory state along a play is determined by the history of states seen so far rather than labels in $\Sigma$.  We refer to \cite{theoretics:9608} for more details on implementing strategies with memory structures. Note that for each $w\in W$, $out(\sigma, w)$ forms an infinite tree of winning plays from $w$. We can think of the subtree as being rooted at $w$: if $w\in V_1$, then this has a single successor $tgt(\sigma(w))$. Otherwise, if $w\in V_2$, this has a successor for each $v\in \delta(w)$. Similarly for any node $v$ which was reached via a partial play $\varsigma\in \delta^+$, if $v\in V_1$ then there is a single successor $\sigma(\varsigma)$, otherwise there is a successor for each $v'\in\delta(v)$. Each node on the tree is reached with a unique sequence of edges, i.e. a partial play $p\in \delta^+$, although since we are dealing with state-labelled arenas we can also specify this as the sequence of states $src(p[0])src(p[1])\ldots src(p[|p|-1])$. We will abuse notation and use $p$ to contextually refer to either a partial play, the equivalent sequence of states, or the partial trace $\pi(p[0])\pi(p[1])\ldots\pi(p[|p|-1])\in \Sigma^+$ generated by the partial play.

    A forest of plays (given a state-labelled 2P game $(V_1,V_2,\delta,\pi)$) is $T = (V^+,e,\tau)$, where $V^+$ is the set of nonempty sequences of states, $e\subseteq V^+\times V^+$ is a relation encoding edges, and $\tau:V^+\rightarrow V$ is a function labelling each sequence of states $u\in V^+$ with the final state $\tau(u) = u[|u|-1]$. The relation $e$ is defined such that for $x,y\in V^+$:
    \begin{itemize}
        \item If $(x,y)\in e$, then $y = xv$ for some $v\in V$, and $(\tau(x),\tau(y))\in \delta$ (all edges must exist in the corresponding game)
        \item If $\tau(x)\in V_1$, then $|e(x)| = 1$ (P1 can choose only one edge on their turn)
        \item If $\tau(x)\in V_2$, then $(\tau(x),\tau(y))\in \delta$ implies $(x,y)\in e$ ($e$ contains all edges which P2 could select on their turn)
    \end{itemize}
    
    We will only be interested in the trees rooted at each $v\in V$, and the subtrees of such trees.
    
    From a P1 strategy $\sigma$, we can generate a forest of plays by adding, for any $x\in V^+$ where $\tau(x)\in V_1$, the edge chosen by $\sigma$ $(x,tgt(\sigma((x[0],x[1])(x[1],x[2]),\ldots)))$ to $e$. Each $v\in V$ is the root of a tree, where the infinite paths through the tree corresponding to infinite paths in $out(\sigma, v)$. If $\sigma$ is winning in states $W\subseteq V$, then for each $w\in W$, each infinite path of the form $\pi(\tau(v_1))\pi(\tau(v_1v_2))\pi(\tau(v_1v_2v_3))\ldots$ is in $L$ where $v_1 = w$ and $(v_i,v_{i+1})\in e$. We will refer to such trees with the same notation as the corresponding sets of traces, i.e. $out(\sigma,w)$. We will refer to the `traces generated by a tree', by which we mean the set of infinite paths through the tree of the above form, taking the $\pi$-labelling at each vertex in the tree. We will talk of subtrees being identical if they are bisimilar under the labelling $\tau$. We will talk of replacing the subtree rooted at $u\in V^+$ with a subtree $T'$ rooted at $u'$, by which we mean redefining the edge relation $e$ for sequences of states extending $u$. First, we remove all edges in $e$ of the form $(ux,y)$ for $x,y\in V^*$. Then, for any $ux\in V^+$ where $x\in V^*$, we add $(ux, uxv)$ to $e$ for each $(u'x,u'xv)\in e$. We will only do this when $\tau(u)=\tau(u')$, which ensures that this is well defined, and that $\tau(ux) = \tau(u'x)$. We will refer to `subtrees rooted at $v$', by which we mean infinite subtrees rooted at some $u\in V^+$ for which $\tau(u) = v$. 

    We also note there is a correspondence between finite-memory strategies and infinite trees of plays which consist of a finite number of subtrees. Suppose we have a memory structure $(M,m_0,\mu)$ and strategy $\sigma:V\times M\rightarrow \delta$. The strategy prescribes moves based on the current memory state and the current state of the arena. Given two subtrees of $out(\sigma,w)$, if the roots of the subtrees are at both the same memory and arena states, then the subtrees will be identical. This means there are $|M|\times |V|$ possible subtrees in $out(\sigma,w)$. On the other hand, if we have an infinite tree of (valid) plays with $n$ different subtrees, note that along every play we must see a repeated subtree somewhere before depth $n+1$. Therefore, the root of every distinct subtree is present somewhere before $n+1$ in the tree. This means we can take the tree up to depth $n$ as the states of our finite-memory strategy $\sigma'$, and for each node in the tree $v$, for each successor $v'$ of $v$, we add a transition pointing to said successor with the label of $v'$. For the nodes at depth $n$, for each of their successors $v'$ in the original tree, we find a point where this subtree is rooted before depth $n$; we can then add a transition in our finite-strategy automaton which points to the pre-depth $n$ root of this subtree. We can pick all other required transitions arbitrarily. Essentially, we have `folded' our tree of plays into a finite automaton. We then set the strategy such that at a memory state rooted at $v\in V_1$, we pick out the corresponding successor. It can be seen that $out(\sigma',w)$ will generate our original tree.
    


    Given a forest of plays, and the tree rooted at some $q\in V$, we will define a process which, for a given state $v\in V_1$, strictly reduces the number of distinct subtrees rooted at $v$. Firstly, we identify along every play, the first point at which $v$ is seen. For each such $v$, we denote by $T$ the subtree rooted at this node, and we denote by $u$ the path to this subtree from our overall tree. We will pick a subtree of $T$, call this $T'$, rooted at $v$ and in which a different successor from the root is chosen than $T$. We will show that within $T$, either we can recursively replace all instances of $T'$ with $T$ and still be left with tree of winning plays overall, or we can recursively replace all instances of $T$ with $T'$ and still be left with a tree of winning plays overall. 

    We have two possible cases; either $T'$ contains $T$ as a subtree, or $T'$ does not contain $T$ as a subtree. Let us denote by $\Lambda_{T'}$ the set of traces generated by infinite paths through the tree $T'$ and by $\Lambda_{T}$ the traces of $T$, and let us denote by $X$ the set of all paths from the root of $T$ to the root of the first instance of $T'$ along any path in $T$. We will first deal with the latter. 
    
     Let us pick a particular $x$. We know that $ux\Lambda_{T'}\subseteq L$. To show that replacing every instance of $T$ with $T'$ still leaves us with a winning tree of plays, we can show $u\Lambda_{T'}\subseteq L$ (essentially replacing the $T$ rooted at $u$ with $T'$). Whereas to show (recursively) replacing every instance of $T'$ with $T$ still leaves us with a winning tree of plays, we can show both $u(X)^n\Lambda_{T}\subseteq L$ and $uX^\omega\subseteq L$ (this covers recursively replacing $T'$ with $T$ at any depth, as well replacing $T'$ with $T$ infinitely). Firstly, we know for any $x\in X$ and $\lambda\in \Lambda_{T'}$ that $ux\lambda\in L$. Therefore we will check each of these traces against condition (1). If for any $x$, we have for all $\lambda$, $u\lambda\in L$, we are done, as $u\Lambda_{T'}\subseteq L$, so now let us assume this is not the case. Otherwise, from (1) we know that for all $x$, $ux^\omega\in L$. From totally ordered residuals, we know for any $x$ either $R_L(u) \supseteq R_L(ux)$ (if the other way around, we would have $u\Lambda_{T'}\subseteq L$), so $ux\Lambda_{T}\subseteq L$, and so for any $x'$, we have $uxx'\Lambda_{T'}\subseteq L$. Similarly we have $R_L(u)\supseteq R_L(uxx')$, and it can be seen that we can repeat this procedure for any sequence of paths in $X$, so $uX^n\Lambda_{T}\subseteq L$ and $uX^n\Lambda_{T'}\subseteq L$. We will now move on to proving the other required part, that we also have $uX^\omega\in L$.
     
     For any sequence $p\in X^n$, we have $up\Lambda_{T'}\subseteq L$, so by condition (1) - since we know we do not have $u\Lambda_{T'}$ - we have $up^\omega\in L$. We will now show that for any $t\in X^*$,$p\in X^n$ we have $utp^\omega$. Since $ut\Lambda_{T'}\subseteq L$ and $up^\omega\in L$, by totally ordered residuals, either $u\Lambda_{T'}\subseteq L$ (which is false by assumption) or $utp^\omega\in L$ as required. 
     
     Now let us take any infinite word $z\in X^\omega$. Let us take a finite $\omega$-semigroup $S=(S_+,S_\omega)$ recognising our language $L$ with a map $\varphi:\Sigma^\infty\rightarrow S$. By the Ramseyan factorization theorem \cite{200475}, the word $z$ admits a factorisation $z=z_1z_2z_3\ldots$ such that $\varphi(z_1z_2\ldots)= \varphi(z_1)\varphi(z_2)\ldots= s\cdot e\cdot e\ldots = se^\omega$ for some $s,e\in S_+$ for idempotent $e$. Observe that by the previous paragraph, $uz_1z_2^\omega\in L$. Since $\varphi(z_1z_2^\omega) = \varphi(z_1)\varphi(z_2)\varphi(z_2)\ldots = se^\omega = \varphi(z)$, then $\varphi(uz) = \varphi(u)\varphi(z)=\varphi(u)\varphi(z_1z_2^\omega)=\varphi(uz_1z_2^\omega)$. So, it must be that $uz_1z_2^\omega\in L$ iff $uz\in L$, so $uz\in L$ for any $z\in X^\omega$. Therefore, we have our other required condition that $uX^\omega\subseteq L$.

    Now let us assume both $T'$ is a subtree of $T$, and $T'$ is a subtree of $T$. Let us further denote by $Y$ the set of all paths from the root of $T'$ to the first instance of a subtree $T$ along some path in $T'$. Showing we can recursively replace every instance of $T$ with $T'$ amounts to showing (both $uX^n\Lambda_{T}\subseteq L$ and $uX^\omega\subseteq L$), whereas the converse amounts to showing (both $uY^n\Lambda_{T'}\subseteq L$ and $uY^\omega\subseteq L$).  
    Our first goal is to show that $u(X^*Y^*)^*\Lambda_{T}\subseteq L$ and $u(X^*Y^*)^*\Lambda_{T'}\subseteq L$. To do this, we will show that given some path $p\in u(X^*Y^*)^*$ where we know $p\Lambda_T\subseteq L$ ($p\Lambda_{T'}\subseteq L$), and so $pX\Lambda_{T'}\subseteq L$ ($pY\Lambda_{T}\subseteq L$) we can remove a layer of $X$-paths ($Y$-paths) and get that $p\Lambda_{T'}\subseteq L$ ($p\Lambda_{T}\subseteq L$). Beginning with the fact that $u\Lambda_T\subseteq L$, repeated applications of this and expansions of $X\Lambda_{T'}\subseteq \Lambda_T$ and $Y\Lambda_{T}\subseteq \Lambda_{T'}$ allow us to get that $u(X^*Y^*)^*\Lambda_{T}\subseteq L$ and $u(X^*Y^*)^*\Lambda_{T'}\subseteq L$. Note that along any path $p\in u(X^*Y^*)^*$ which leads to subtree rooted at $T$, we can apply the previous part of the proof to get that either we have $p(X)^n\Lambda_{T}\subseteq L$ and $pX^\omega\subseteq L$, or we have $p\Lambda_{T'}\subseteq L$. In the latter case we are done, so let us consider what happens if $p(X)^n\Lambda_{T}\subseteq L$ and $pX^\omega\subseteq L$. By totally ordered residuals, since $uX\Lambda_{T'}\subseteq L$, either $uX^\omega\subseteq L$ or $p\Lambda_{T'}$. In the former case, we now have the conditions required to recursively replace all $T'$ subtrees with $T$, so we do not to continue with the proof. Otherwise, we have $p\Lambda_{T'}$ as required. Similarly, with $T$ and $T'$, $X$ and $Y$ swapped.
    

     We can use this to replace a layer of $T$ trees with $T'$ trees at any depth, and vice versa. Let us note that $(XY)^\omega\subseteq \Lambda_{T}$ and $(YX)^\omega\subseteq\Lambda_{T'}$, so we have $u(X^*Y^*)^*(XY)^\omega\subseteq L$ and $u(X^*Y^*)^*(YX)^\omega\subseteq L$. From this, we can show our final required condition that either $uX^\omega\subseteq L$ or $uY^\omega\subseteq L$. Suppose this is not the case, so there is a $a\in X^\omega$ and $b\in Y^\omega$ such that $ua\notin L$ and $ub\notin L$. Again, we can find factorisations $a=a_{1}a_2a_3\ldots$ such that $\varphi(a)= \varphi(a_1)\varphi(a_2)\ldots= s\cdot e\cdot e\ldots = se^\omega = \varphi(a_1a_2^\omega)$ for some $s,e\in S_+$ for idempotent $e$ and $b=b_{1}b_2b_3\ldots$ such that $\varphi(b)= \varphi(b_1)\varphi(b_2)\ldots= s'\cdot e'\cdot e'\ldots = s'e'^\omega = \varphi(b_1b_2^\omega)$ for some $s',e'\in S_+$ for idempotent $e'$. So, we know $ua_1a_2^\omega\notin L$ and $ub_1b_2^\omega\notin L$. However, note that both $ua_1(a_2b_2)^\omega\in u(X^*Y^*)^*(XY)^\omega\subseteq L$ and $ub_1(a_2b_2)^\omega\in u(X^*Y^*)^*(XY)^\omega\subseteq L$. So, we can apply condition (2) to get either $ua_1a_2^\omega\in L$, which we assumed was not the case, or $ua_1b_2^\omega\in L$. Similarly, we are forced to have $ub_1a_2^\omega\in L$. Applying totally ordered residuals gives us either $ua_1a_2^\omega\in L$ or $ub_1b_2^\omega\in L$, in either case contradicting our initial assumption. Therefore, we have either $uX^\omega\subseteq L$ or $uY^\omega\subseteq L$ as required: if the former, within the tree rooted at path $u$ we can recursively replace all subtrees $T'$ with $T$, if the latter we can recursively replace all subtrees $T$ with $T'$.
    
    
    
    With the above two methods, we can replace all instances of $T$ in a subtree with $T'$ or vice versa, for two subtrees $T$ and $T'$ rooted at the same node.
    We repeat this for the $T$ rooted at the first instance of $v$ along every path; now at every subtree rooted at $v$, we have removed a positionality violation from within that subtree. We then apply totally ordered residuals to replace every tree rooted at a first instance of $v$ along a path with the same tree (this will make the next step simpler). We will call the overall tree at this point $T^q_{new}$, and suppose we have recursively replaced every subtree $T'$ of $T$ with $T$ itself. We first show that $T^q_{new}$ can be implemented by a finite-memory strategy; essentially, the goal is to track when we have reached $v$ for the first time, and then ensure that whenever we would be generating the subtree $T'$ we are in fact generating the subtree $T$. First we take our original memory structure $(M,m_0,\mu)$ and strategy $\sigma: V\times M \rightarrow V$ implementing our previous tree. We will construct a new memory structure which implements $T^q_{new}$. Suppose $m_T\in M$ is the memory state corresponding to the subtree $T$, and $m_{T'}$ is the memory state corresponding to the subtree $T'$. To track when we first reach the state $v$, we need to create 2 copies $M_1$ and $M_2$ of $M$; we will be in $M_1$ prior to reaching $v$ and $M_2$ after reaching $v$ for the first time.

    More formally, we replace our memory structure with $(M_1\sqcup M_2,(m_0,1),\mu')$ where:
    \begin{itemize}
        \item $M_1\sqcup M_2=\{(m,1)\mid m\in M\}\sqcup\{(m,2)\mid m\in M\}$
        \item $(m_0,1)\in M_1$ is the initial memory state.
        \item $\mu':M_1\sqcup M_2\times V\rightarrow M_1\sqcup M_2$ is defined as follows: $$\mu'((m,1),s) = \begin{cases}
            (\mu(m,s),1) & \text{ if $s\neq v$} \\
            m_T & \text{ otherwise}
        \end{cases}$$ 
        $$\mu'((m,2),s) = \begin{cases}
            (\mu(m,s),2) & \text{ if $\mu(m,s) \neq m_{T'}$} \\
            m_T & \text{ otherwise}
        \end{cases}$$
    \end{itemize}

    It can be seen that this ensures that the first subtree rooted at $v$ along any path is $T$, and that all $T'$ subtrees within $T$ have been replaced by $T$. It can be seen that this procedure preserves winning trees of plays, as proven above. Note that this strictly reduces the number of subtrees rooted at $v$ which select different successor states. This is because if we are at a state $v$, we must be in a memory state in $M_2$, and the memory state $(m_{T'},2)$ is now unreachable. So, the number of memory states reachable when at $v$ is reduced by at least 1. This continues if we repeat this process for further choices of $T$ and $T''$ rooted at $v$ (although the total space of memory states increases), so we can continue this process until we are left with a forest of plays in which the same successor state is chosen for every subtree rooted at $v$. Once this is the case, repeating this process for other states will never add different successor states for $v$. Therefore, once we repeat this for every $v\in V_1$, we have a collection of trees in which every P1 state is followed by the same choice of successor state, from which it is simple to construct a positional strategy.

\end{proof}

\stateedgepos*
\begin{proof}
    Suppose we have an $L$ which can be played optimally with memory of the previous colour in edge-labelled games. Take an arbitrary state-labelled game $S = (V_1,V_2,\delta,\pi)$. Take a memoryful strategy $\sigma$ in $S$ for P1, with winning states $W\subseteq V$. Take the edge-labelled game $E = (V_1, V_2\cup\{*\},  \delta')$ where $\delta' = \{(v,\pi(v'),v')| (v,v')\in \delta\} \cup \{(*,\pi(w),w)|w\in W\}$.  Given a strategy $\sigma$ in $S$, to obtain a strategy in $E$ we set $\sigma'(*q_1\ldots q_n) = \sigma'(q_1\ldots q_n) = (q_n,\pi(\sigma(q_1\ldots q_n)),\sigma(q_1\ldots q_n))$. It can be seen that if $\sigma$ in $S$ is winning from each $w\in W$, then $\sigma'$ is winning from $*$ in $E$. Since $L$ can be played optimally with a memory of the previous colour, we have a $\varsigma': V\cup (V\times\Sigma)\rightarrow V$ which is winning from $*$ when provided the previous colour seen in the play. We define a positional strategy in our state-labelled arena $\varsigma:V\rightarrow V$ such that $\varsigma(v)=\varsigma'(v,\pi(v))$. It can be seen that from a $w\in W$, all traces generated by $\varsigma$ are also generated by $\varsigma'$ from $*$ - since all traces generated by $\varsigma'$ from $*$ are winning, $\varsigma$ will be winning from all $w\in W$.



    For the other direction, suppose we have a state-labelled positional $L$. Take an arbitrary edge-labelled game $E = (V_1,V_2, \delta)$. Let $\delta_1 = \{(c,v')| (v,c,v')\in\delta, v'\in V_1\}$ and $\delta_2 = \{(c,v')| (v,c,v')\in\delta, v'\in V_2\}$. Let us take the state-labelled game $S = (\delta_1,\delta_2, \delta', \pi)$ where: $((c,v'),(c',q'))\in \delta'$ iff $(v',c',q')\in \delta$; and $\pi((c,v))=c$. Suppose we have a memoryful strategy $\sigma$ in $E$ which is winning at states $W\subseteq V$. From this, we can obtain a memoryful strategy $\varsigma$ over $S$ which is winning at, for each $w\in W$ the state $(c,v')$ where $\sigma(w) = (w,c,v')$. To do this, first for each $(c,v)$ we will select, if it exists, a state $v_{(c,v)}$ which is winning under $\sigma$ and is such that $\sigma(v_{(c,v)}) = (v_{(c,v)},c,v)$. If this does not exist, we select an arbitrary state $v_{(c,v)}$ with an outgoing edge $(v_{(c,v)},c,v)$ Then, we set $\varsigma((c_1,v_1)\ldots(c_n,v_n)) = (c_{n+1},v_{n+1})$ where $\sigma((v_{(c_1,v_1)},c_1,v_1)\ldots (v_{n-1},c_n,v_n)) = (v_n,c_{n+1},v_{n+1})$. This simply copies the edges chosen in winning plays where the first edge chosen is $(w,c,v')$ for some $w$, so can be seen to winning at any $(c,v')$ where either $(w',c,v')$ is chosen as the first edge from a winning play from some $w'\in V_1$, or where $(w',c,v)\in\delta$ for some winning $w'\in V_2$. Since $L$ is state-labelled positional, there will be a positional strategy, which is a function $\sigma':\Sigma\times V \rightarrow \Sigma\times V$, with the same winning states. From this we can obtain a function $\sigma'':(\Sigma\times V) \cup V \rightarrow \delta$ by, for each $c,v$ such that $\sigma'(c,v) = (c',v')$, setting $\sigma''(c,v) = (v,c',v')$. This will be defined for any $c,v$ pair where $c$ is a label on an outgoing edge from $v$. We will set the action of $\sigma''(v)$ to either: if there is state $(c,v')$ in $S$ for which $\sigma'$ is winning and such that $(v,c,v')\in\delta$, then select this, or otherwise, select an arbitrary outgoing edge from $v$. If we generate a play from $w\in W$ using $\sigma''$ in $E$, by providing the previously seen label, it can be seen that this will still be winning. Therefore, whenever $L$ is state-positional, for any memoryful strategy in an edge-labelled game we can find a strategy with memory of the previously seen label that has the same winning states.
\end{proof}

\preindvariety*
\begin{proof}
    First we will show closure under boolean operations. Suppose $x\in L_1\cap L_2$ for prefix independent $L_1,L_2$; this is the case iff $x\in L_1$ and $x\in L_2$ iff $ux\in L_1$ and $ux\in L_2$ iff $ux\in L_1\cap L_2$. Similarly for union. For complement, $x\in \overline{L}$ iff $x\notin L$ iff $ux\notin L$ iff $ux\in \overline{L}$.
    For prefix-independent $\omega-$languages, $R_L(u)=L$ for any $u\in \Sigma^*$, so they are closed under taking residuals.
    Given a map $\varphi:(A^+,A^\omega)\rightarrow (B^+,B^\omega)$ and a prefix-independent $L\subseteq B^\omega$, we will show $\varphi^{-1}(L)$ is prefix-independent. Suppose $x\in \varphi^{-1}(L)$ for some $x\in A^\omega$. Then, $\varphi(x)\in L$, so for all $u\in B^+$, $u\varphi(x)\in L$. This includes all $u\in \varphi(A^+)$, so for all $u\in A^+$, $\varphi(ux)\in L$, so $ux\in \varphi^{-1}(L)$. For the other direction, suppose for some $u\in A^+,x\in A^\omega$ we have $ux\in \varphi^{-1}(L)$, i.e. $\varphi(ux)\in L$, so $\varphi(u)\varphi(x)\in L$, and by prefix independence, $\varphi(x)\in L$. So, $x\in \varphi^{-1}(L)$ as required.

    Suppose an $\omega-$semigroup $(S_+,S_\omega)$ satisfies the identity $uv = v$. Take some map $\varphi:(A^+,A^\omega)\rightarrow (S_+,S_\omega)$, and some $F\subseteq S_\omega$. For any $u\in A^+,x\in A^\omega$, we have $\varphi(ux) = \varphi(u)\varphi(x)=\varphi(x)$, so $x\in \varphi^{-1}(F)$ iff $ux\in\varphi^{-1}(F)$. 
    
    Suppose we have a prefix-independent $L$, and we take the syntactic semigroup. For infinite words $x,y\in A^\omega$, then $x \sim_L y$ iff for all $u\in A^+$, $ux\in L$ iff $uy\in L$, which by prefix independence means $x\in L$ iff $y\in L$. For any $uv \in A^\omega$, by prefix independence $uv \sim_L v$, so in the syntactic semigroup the identity $uv = v$ holds.
\end{proof}

\preindtrivvariety*
\begin{proof}
    It is well known that $\infty-$varieties of languages correspond to varieties of $\omega-$semigroups. In other words, given an $\infty-$variety $\mathcal{V}$ there is a variety of $\omega$-semigroups $\mathbf{V}$ such that for any $V\in \mathbf{V}$ and $\varphi:\Sigma^+\rightarrow V$, and any $F\subseteq V$, we can find the language $\varphi^{-1}(F)$ in $\mathcal{V}$.

    The languages $\{\emptyset,\Sigma,\Sigma^\omega,\Sigma^\infty\}$ are all prefix-independent and positional, and form an $\infty$-variety. This can be generated by the trivial variety consisting of $\omega-$semigroup $(\{\star\},\{\star\})$.

    Suppose we have an $\infty-$variety consisting entirely of prefix-independent positional languages. We need that every language recognised by the corresponding variety of $\omega$-semigroups is positional and prefix-independent. A language $L$ is prefix independent when $v\in L$ iff $uv\in L$. If we have an $\omega$-semigroup $S=(S_+,S_\omega)$ and a set $F\subseteq S$ recognising a prefix-independent language $L$, then we require that for all $u\in S_+$, $v\in S_\omega$, $uv^\omega\in F$ iff $v^\omega\in F$. Since we could have any $F\subseteq S_\omega$, including the set $\{uv^\omega\}$, we require that the $S$ satisfies the equation $uv^\omega=v^\omega$. For positionality, we know we require $(uv)^\omega\in L$ implies either $u^\omega\in L$ or $v^\omega\in L$, so we need that in S for each $u,v\in S_+$ either $(uv)^\omega=u^\omega$ or $(uv)^\omega=v^\omega$.

    Suppose our variety contains an $\omega$-semigroup $S$ where $u^\omega\neq v^\omega$ for some $u,v\in S_+$. Every element in $S_\omega$ is equal to $t^\omega$ for some $t\in S_+$, so this is required to have more than one element in $S_\omega$ and recognise more $\omega$-languages than $\emptyset$ and $\Sigma^\omega$. Varieties are closed under finite products, so let us look at the product $S\times S$, and assume it recognises only prefix-independent and positional languages. In $S\times S$, $((u,v)\cdot(v,u))^\omega$ either equals $(u,v)^\omega=(u^\omega,v^\omega)$ or $(v,u)^\omega=(v^\omega,u^\omega)$ by positionality. $((u,v)\cdot(v,u))^\omega = (uv,vu)^\omega = ((uv)^\omega,(vu)^\omega)$ by the definition of multiplication and $\omega$ on the product. For any $s,t\in S_+$, $(st)^\omega=s(ts)^\omega=(ts)^\omega$, so $((uv)^\omega,(vu)^\omega)=(x,x)$ for some $x\in S_\omega$. However, we already know $((u,v)\cdot(v,u))^\omega$ equals either $(u^\omega,v^\omega)$ or $(v^\omega,u^\omega)$ where $u^\omega\neq v^\omega$, giving us a contradiction. Therefore this product does not recognise only positional and prefix-independent languages. Any variety of languages recognising more $\omega-$languages than $\emptyset$ and $\Sigma^\omega$ will have at least one $\omega-$semigroup in the corresponding variety where $u^\omega\neq v^\omega$ for some $u,v\in S_+$. Varieties of $\omega-$semigroups are closed under products, so a variety of languages that contains only prefix-independent and positional languages must only recognise the $\omega$-languages $\emptyset$ and $\Sigma^\omega$.
\end{proof}

\nogo*
\begin{proof}
    Suppose a variety $\mathbf{V}$ containing only positional languages contained a prefix-independent language $L$. We know by \Cref{prop:prefindvariety} that prefix-independent languages form a variety, so the variety generated by $\{L\}$ would be a subvariety of this, and would be a subvariety of $\mathbf{V}$ so would contain only positional languages. Therefore, it would be a variety consisting only of prefix-independent and positional languages, which we know to be impossible by Proposition \ref{prop:prefindimpossible}.
\end{proof}

\fgvariety*
\begin{proof}
    Clearly this is closed under finite union by definition. Suppose we have a finite intersection of languages of the form $\bigcap_{i\in I}FG(A_i)$; we will show this can be written as $FG(B)$ for some $B\subseteq \Sigma$. Any word $w\in\bigcap_{i\in I}FG(A_i)$ has a finite index $i$ beyond which all characters are in $\bigcap_{i\in I} A_i$, so $w\in FG(\bigcap_{i\in I} A_i)$. If $w\in FG(\bigcap_{i\in I} A_i)$, then clearly it is in each $FG(A_i)$ and so $w\in\bigcap_{i\in I}FG(A_i)$. Given a finite intersection of languages of the form $FG(A_1) \cup \ldots \cup FG(A_n)$, we can distribute the intersection through the union and apply the previous step to obtain a language of the form $FG(B_1) \cup \ldots \cup FG(B_m)$.
    These languages are prefix independent, so are trivially closed under taking residuals.
    Suppose we have a language $L\subseteq B^\omega$ of this form, and a map $\varphi: A^\infty \rightarrow B^\infty$. We will show $\varphi^{-1}(L)$ is also of this form. Given a set $S\subseteq \Sigma^+$ we denote by $col(S)$ the smallest set $C\subseteq \Sigma$ such that $S\subseteq C^+$. It can be seen that $\varphi^{-1}(FG(A_1)\vee\ldots\vee FG(A_n)) = FG(col(\varphi^{-1}(A_1^+))\vee\ldots\vee col(\varphi^{-1}(A_n^+)) )$. For the left-to-right inclusion, we have some $w\in A^\omega$ such that at some point $j$, at all $k\geq j$, $\varphi([k])$ only contains labels in some $A_i$, so $\varphi(w[k])\in A_i^+$. So, $w\in FG(col(\varphi^{-1}(A_i^+)))$. For the other direction, we have some $w\in A^\omega$ such that at some point $j$, at all $k\geq j$, $w[k]\in col(\varphi^{-1}(A_i^+))$. So, $\varphi (w[k])\in A_i^+$, therefore $\varphi(w)\in FG(A_1)\vee\ldots\vee FG(A_n)$ as required.
\end{proof}

\positionalformulas*
\begin{proof} $(GFp \wedge FGq)$
    
    It can be seen that $GFp \wedge FGq$ is prefix-independent. Let us take trace a of the form $u(vw)^\omega$ satisfying $GFp\wedge FGq$. We know both $v$ and $w$ must contain only labels containing $q$, else $u(vw)^\omega$ would not satisfy $FGq$. Similarly, either $v$ or $w$ must contain $p$. If $v$ contains $p$, then $uv^\omega$ satisfies the property, otherwise $uw^\omega$ satisfies the property. Therefore, by \Cref{prop:pre_ind_conditions}, $GFp \wedge FGq$ is positional (over edge and state-labelled games).

    ($G(pUq)$)
    
    It can be seen the language has only 2 distinct sets of residuals: $\emptyset$ and $G(pUq)$. At any point on a trace satisying $G(pUq)$, either $p$ or $q$ is in the current label, and $q$ is seen within a finite number of steps. Take a satisfying trace of the form $uvwx^\omega$. The trace $uwx^\omega$ will satisfy $G(pUq)$, as there must be a $q$ within finite steps of the start of $wx^\omega$. Take a satisfying trace of the form $u(vx)^\omega$. Either $v$ or $x$ contains a label with $q$ - if $v$, then $uv^\omega$ satisfies the property, otherwise $ux^\omega$ satisfies the property.

    ($GFp\land FGq\land Gr$)
    
    It can be seen the language has only 2 distinct sets of residuals: $\emptyset$ and $GFp \wedge FGq \wedge Gr$. At any point on a trace satisfying $GFp \wedge FGq \wedge Gr$, $r$ is in the current label, $p$ is seen within a finite number of steps, and $q$ is seen in every label within a finite number of steps. Take a satisfying trace of the form $uvwx^\omega$: the trace $uwx^\omega$ must satisfy $GFp \wedge FGq \wedge Gr$, as there must be a $p$ within finite steps of the start of $wx^\omega$, there is still a finite point after which $q$ is seen in every label, and all positions are still labelled with $r$. Take a satisfying trace of the form $u(vx)^\omega$. Either $v$ or $x$ contains a label with $p$ - if $v$, then $uv^\omega$ satisfies the property, otherwise $ux^\omega$ satisfies the property.

    $(p \vee Xp)$
    Let $L=G(p \vee Xp)$. Suppose we have $uvwx^\omega\in L$, with $uwx^\omega\notin L$. Then, $u$ must end in a set $X$ not containing $p$ and $w$ must begin in a set $Y$ not containing $p$. Therefore, $v$ must begin and end with sets containing $p$ and have no substrings longer than 1 consisting of sets not containing $p$, or $uvwx^\omega$ would not have been in the language. So, $uv^\omega\in L$. Similar for words of the form $u(vwx)^\omega$. There are 3 different residuals. For words $w\in \Sigma^*$ which have a run of two or more sets not containing $p$, $R_L(w) = \emptyset$. If this is not the case, then if $w$ is the empty word or ends in a set containing $p$ then $R_L(w) = L$. Otherwise if $w$ ends in a set not containing $p$ then $R_L(w) = p \wedge XL$. Since $\emptyset \subseteq p \wedge XL \subseteq L$, then the residuals are totally ordered. Therefore, $L$ is positional.
\end{proof}

\mutualexpress*
\begin{proof}[Proof sketch]
    We prove the claim in the one-agent case,
    the claim for full ATL can be proven is an analogous way,
    generalising from paths to strategies.
    
    We have the following logical equivalence
    \[E(GF\psi\land FG\varphi\land G\phi)\equiv E(\phi U (E(G((\varphi\land\phi)
    U(\psi\land\varphi\land\phi)))\]
    
    ($\Rightarrow$)
    Assume we have a model $M$ and a state $x$ which
    satisfies the left formula.
    This means that there is a trace $\lambda$ starting at $x$ which
    satisfies $GF\psi$, $FG\varphi$ and $G\phi$.
    We show that $\lambda$ also satisfies the right formula.
    From the semantics of $FG\varphi$, there is an $n\in\mathbb{N}$
    such that $\lambda[n+i]\vDash\varphi$ for all $i\in\mathbb{N}$.
    It thus suffices
    to show
    $\lambda[n]\vDash E(G((\phi\land\varphi)U(\psi\land\varphi\land\phi)))$.
    We choose $\lambda[n..]$ as our witnessing trace,
    so are left to verify $\lambda[n..]\Vdash G((\phi\land\varphi)U(\psi\land\varphi\land\phi))$,
    to do so we fix an arbitrary $i\in\mathbb{N}$ and
    check $\lambda[n+i..]\Vdash(\phi\land\varphi)U(\psi\land\varphi\land\phi)$.
    From $\lambda$ satisfying $GF\psi$, there will exist
    some $i'\in\mathbb{N}$ with $\lambda[n+i+i']\vDash\psi$,
    we already know that every point along $\lambda$ satisfies
    $\phi$ and all points after $n$ satisfy $\varphi$,
    so are done.

    ($\Leftarrow$) Assume we have a model $M$ and a state $x$
    which satisfies the right formula.
    There is a trace $\lambda$ starting at $x$ with
    some $n$ s.t. $\lambda[0..n-1]\Vdash\phi$
    and with $\lambda[n]\vDash EG((\varphi\land\phi)U(\psi\land\varphi\land\phi))$.
    Hence we have a trace $\lambda'$ starting at $\lambda[n]$
    and,
    for all $i\in\mathbb{N}$, there exists
    some $i'\in\mathbb{N}$ with $\lambda'[i..i+i'']\vDash\varphi\land\phi$
    for all $i''<i'$
    and $\lambda[i+i']\vDash\psi\land\varphi\land\phi$.
    We can see that
    $\lambda[0..n-1]\cdot\lambda'\Vdash GF\psi\land FG\varphi\land G\phi$
    holds from the previous observations,
    so the model $M$ satisfies the left formula at $x$.

    Furthermore, we can express $GU$ with $GF\land FG\land G$:

    $E(G\psi U \varphi) \equiv E(GF\varphi \wedge FG\top \wedge G(\psi \vee \varphi))$

    ($\Rightarrow$)
    Assume we have a model $M$ and a state $x$ which
    satisfies the left formula.
    This means that there is a trace $\lambda$ starting at $x$ which
    for $\lambda[i..]$ for all $i\geq 0$, satisfies $\psi U \varphi$.
    We show that $\lambda$ also satisfies the right formula.
    At any $i$, there must be a finite point $j\geq i$ such that $\lambda[i..]\Vdash \varphi$,
    so $GF\varphi$ must be satisfied. $FG\top$ is trivially satisfied. 
    At all $i$, $\lambda[i] \Vdash \psi U \varphi$, which means
    $\lambda[i] \Vdash \varphi \vee (\psi \wedge X(\psi U \varphi))$. Since at all $\lambda[i]$
    either $\psi$ or $\varphi$ holds, then $G(\psi \vee \varphi)$ is satisfied. Therefore,
    $x\vDash E(GF\varphi \wedge FG\top \wedge G(\psi \vee \varphi))$.

    ($\Leftarrow$) Assume we have a model $M$ and a state $x$
    which satisfies the right formula. There is a trace $\lambda$ starting at $x$
    s.t. $\lambda\Vdash GF\varphi \wedge FG\top \wedge G(\psi \vee \varphi)$. We must show
    that for each $i\geq 0$, $\lambda[i..]\Vdash \psi U \varphi$. This means that for some $j\geq i$, $\lambda[j..]\Vdash \varphi$
    and for all $i\leq k < j$, $\lambda[k..]\Vdash \psi$. Let us take the smallest $j\geq i$ such that $\lambda[j..]\Vdash \varphi$:
    we know this must exist as $\lambda \Vdash GF\varphi$. We therefore know for all $i \leq k < j$ that $\lambda[k..]\nVdash \varphi$, but also that
    $\lambda\Vdash G(\psi \vee \varphi)$, so it must be that for each $k$, $\lambda[k..]\Vdash \psi$. So, each $\lambda[i..]$ satisfies $\psi U \varphi$
    and therefore $\lambda$ satisfies $G\psi U \varphi$ as required.
\end{proof}

\express*
\begin{proof}[Proof sketch]
    Given an agent $i\in{\sf Ag}$,
    we define a construction $(-)^{ATL}_i$ which transforms a rooted transition system into a rooted concurrent game structure.
    Given a rooted transition system $\mathcal{M}$, $\mathcal{M}_i^{ATL}$ will have the same set of states $X$ and assign the same atomic propositions to them. For each state $x\in X$, the set of actions available to agent $j$ will be $\{*\}$ for $i\ne j$. The set of actions available to agent $i$ will be one-to-one correspondence with the successors at $x$ in $\mathcal{M}$.
    The transition function of the concurrent game structure will assign each of tuple $(x,(a,*,\dots,*))$ to a different successor in $\mathcal{M}$ at $x$.
    It is routine to prove the following fact about this construction:

    \begin{equation}
        \label{eq:fact}
    \mathcal{M}_i^{ATL}\Vdash\langle C\rangle\psi\iff\begin{cases}
    \mathcal{M}\Vdash\exists\psi & \text{if }i\in C\\
    \mathcal{M}\Vdash\forall\psi & \text{if }i\not\in C
    \end{cases}
    \end{equation}

    We now prove the statement of the lemma.
    By \cite[Lemma 10.3.2]{Demri_Goranko_Lange_2016},
    it suffices to show that $(\mathcal{M}_n)^{ATL}_i\vDash\langle i\rangle\varphi$ and $(\mathcal{N}_n)^{ATL}_i\nvDash\langle i\rangle\varphi$ (which easily follows as we have $\mathcal{M}_i^{ATL}\vDash\langle i\rangle\varphi\iff\mathcal{M}_i\Vdash\exists\varphi$),
    and that $(\mathcal{M}_n)^{ATL}_i\vDash\psi\iff(\mathcal{N}_n)^{ATL}_i\vDash\psi$
    for all $ATL+A$ formulas $\psi$ which are of modal depth less than $n$.
    We can do induction on the structure of $\psi$. The semantics of atomic propositions and path formula are left unchanged
    by $(-)^{ATL}_i$, the only non-trivial case is $\langle C\rangle\psi$,
    this follows easily once we have established \Cref{eq:fact}.
    
\end{proof}

\eatltogffg*
\begin{proof}
     We employ \Cref{lem:express_transfer}, and find two classes of rooted transition systems which are indistinguishable by any $ECTL$ formula, but for which $E(GF\varphi\land FG\psi)$ is a separator.

     This follows from the proof of Lemma 2.8.4 in \cite{larou1994}, which defines two classes of models $(M_i)_{i\in \mathbb{N}},(N_i)_{i\in \mathbb{N}}$ (Figure \ref{fig:GF-FG-separator}) that are indistinguishable by any ECTL formulae but for which $E(GFp\wedge Gq)$ is a separator. However, it can be seen that $M_i,a_i\vDash E(GFa \wedge FGb)$ and $N_i,\alpha_i\nvDash E(GFa \wedge FGb)$, so $E(GFa \wedge FGb)$ is also not expressible in ECTL.

     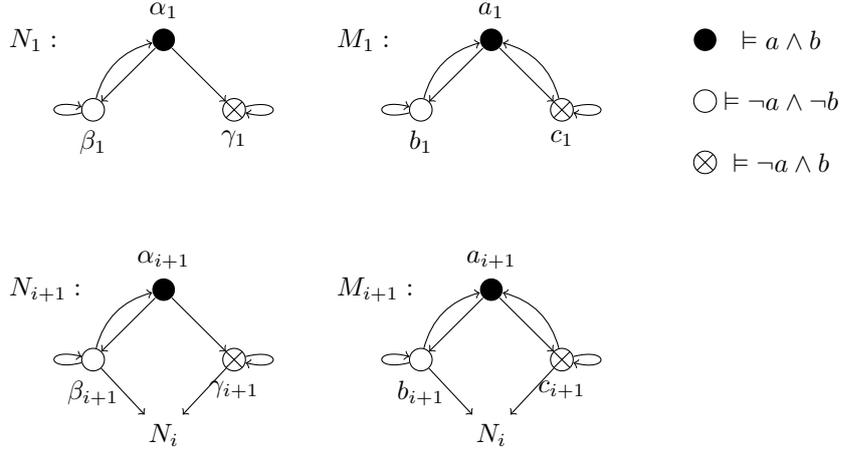
\begin{figure}
    \centering
    \begin{tikzpicture}[node distance=1cm, 
    roundnode/.style={circle, draw, inner sep=3pt},
    dotnode/.style={circle, fill=black, inner sep=3pt},
    cross/.style={path picture={\draw[black] (path picture bounding box.south east) -- (path picture bounding box.north west) (path picture bounding box.south west) -- (path picture bounding box.north east);}},
    crossnode/.style={circle, draw, cross, inner sep=3pt}]

\node[dotnode, label=above:$\alpha_1$] (a1) {};
\node[roundnode, below left=of a1, label=below:$\beta_1$] (b1) {};
\node[crossnode, below right=of a1, label=below:$\gamma_1$] (g1) {};

\draw[->] (a1) -- (b1);
\draw[->] (a1) -- (g1);
\draw[->] (b1) to[bend left] (a1);
\draw[->] (b1) to[loop left] (b1);
\draw[->] (g1) to[loop right] (g1);

\node[left=2cm of a1, anchor=west] (N1label) {$N_1:$};

\node[dotnode, right=4cm of a1, label=above:$a_1$] (m1) {};
\node[roundnode, below left=of m1, label=below:$b_1$] (n1) {};
\node[crossnode, below right=of m1, label=below:$c_1$] (o1) {};

\draw[->] (m1) -- (n1);
\draw[->] (m1) -- (o1);
\draw[->] (n1) to[bend left] (m1);
\draw[->] (o1) to[bend right] (m1);
\draw[->] (n1) to[loop left] (n1);
\draw[->] (o1) to[loop right] (o1);

\node[left=2cm of m1, anchor=west] (M1label) {$M_1:$};

\node[dotnode, below=3cm of a1, label=above:$\alpha_{i+1}$] (an1) {};
\node[roundnode, below left=of an1, label=below:$\beta_{i+1}$] (bn1) {};
\node[crossnode, below right=of an1, label=below:$\gamma_{i+1}$] (gn1) {};

\draw[->] (an1) -- (bn1);
\draw[->] (an1) -- (gn1);
\draw[->] (bn1) to[bend left] (an1);
\draw[->] (bn1) to[loop left] (bn1);
\draw[->] (gn1) to[loop right] (gn1);

\node[below=1.5cm of an1] (Ni) {$N_i$};
\draw[->] (bn1) -- (Ni);
\draw[->] (gn1) -- (Ni);

\node[left=2cm of an1, anchor=west] (Nn1label) {$N_{i+1}:$};

\node[dotnode, right=4cm of an1, label=above:$a_{i+1}$] (am1) {};
\node[roundnode, below left=of am1, label=below:$b_{i+1}$] (bm1) {};
\node[crossnode, below right=of am1, label=below:$c_{i+1}$] (cm1) {};

\draw[->] (am1) -- (bm1);
\draw[->] (am1) -- (cm1);
\draw[->] (bm1) to[bend left] (am1);
\draw[->] (cm1) to[bend right] (am1);
\draw[->] (bm1) to[loop left] (bm1);
\draw[->] (cm1) to[loop right] (cm1);

\node[below=1.5cm of am1] (Mi) {$N_i$};
\draw[->] (bm1) -- (Mi);
\draw[->] (cm1) -- (Mi);

\node[left=2cm of am1, anchor=west] (Mn1label) {$M_{i+1}:$};

\node[dotnode, right=2.5cm of m1] (dotpropositions) {};
\node[roundnode, below=0.5cm of dotpropositions] (roundpropositions) {};
\node[crossnode, below=0.5cm of roundpropositions] (crosspropositions) {};

\node[right of =dotpropositions] {$\vDash a \wedge b$};
\node[right of =roundpropositions] {$\vDash \neg a \wedge \neg b$};
\node[right of =crosspropositions] {$\vDash \neg a \wedge b$};

\end{tikzpicture}
    \caption{Recreation of Figure 2.2 from \cite{larou1994}}
    \label{fig:GF-FG-separator}
\end{figure}
\end{proof}

\gffgtogu*
\begin{proof}
     We will show the corresponding statement for CTL, which is sufficient by Lemma \ref{lem:express_transfer}. We will use the model in Figure \ref{fig:GUseparator}, which we will refer to as $\mathcal{M}$. This model is defined such that at any state $s'_i$, the formula $EG(pUq)$ holds, but at any state $s_i$, the same formula does not hold. Therefore, to prove the proposition it is sufficient to show that on this model, for all $n\in \mathbb{N}$ and $\varphi\in ECTL + E(GF\psi \wedge FG\varphi)$ such that $|\varphi|\leq n$, it is the case that $s_n \vDash \varphi$ iff $s'_n\vDash\varphi$. We will assume states have labels in $2^{\{p,q\}}$.

    Firstly, we will explain how $\mathcal{M}$ is defined. The formula $p\wedge\neg q$ holds at all $s_i$ and $s'_i$. The formula $\neg p \wedge \neg q$ holds at all states $l$. The state $s_0$ has a single transition leading to itself. All other $s_i$ and $s'_i$ are connected through sequences of states. Each $\sigma_i$ represents a set of possible paths from $s_i$ to itself. A path $\varsigma$ is in $\sigma_i$ iff it passes through a sequence of states of length $\leq i$, and is such that $\varsigma\nvDash pUq$. For example, $s_3$ will be connected to itself through a sequence of transitions between states $s_3\rightarrow\varsigma_1\rightarrow\varsigma_2\rightarrow\varsigma_3\rightarrow s_3$ such that for each $\varsigma_i$, $\varsigma_i\vDash p\wedge\neg q$. A path $\varsigma$ is in $w_i$ iff it passes through a sequence of states of length $\leq i$, and is such that $\varsigma\vDash pUq$. For example, $w_2$ contains transitions through states $s_2\rightarrow\varsigma_1\rightarrow\varsigma_2\rightarrow s_1$. Each $\sigma'_i$ contains all $\Sigma_{1\leq j\leq i}(2^{\{p,q\}})^j$ paths of length $\leq i$. Note that each $s_i$ and each $s'_i$ is connected to $s_{i-1}$ via $w_i$ paths. 

    The idea is that $CTL + E(GF\psi \wedge FG\varphi)$ cannot distinguish $E (FG(\psi U \varphi))$ from $EG(\psi U \varphi)$.

    Note that for $i>0$, for each $s_i$ and $s'_i$, that $s'_i\vDash \exists G(pUq)$ (by taking paths in $\sigma_i$ infinitely often) but $s_i\nvDash \exists G(pUq)$, since from an $s_i$, to realise $pUq$ infinitely often we much reach a state $s'_i$ which requires going through some $l$ in which $\neg p \wedge \neg q$ holds, breaking the possibility of $G(pUq)$. 

    We can prove inductively that for any $CTL + GF\wedge FG$ formula $\varphi$ of size $|\varphi| = n$, we have $s_{n_1}\vDash \varphi$ iff $s'_{n_2}\vDash\varphi$ for any $n_1,n_2\geq n$. Additionally, we will show for a path $\varsigma'\in \sigma'_n$, if $\varsigma'\nvDash p U q$ then there is a path of the same length $\varsigma\in \sigma_n$ such that for each $0\leq i \leq |\varsigma|$, $\varsigma_i\vDash \varphi$ iff $\varsigma'_i\vDash \varphi$, and if $\varsigma'\vDash p U q$ then there is a path of the same length $\rho\in w_n$ such that for each $0\leq i \leq |\rho|$, $\rho_i\vDash\varphi$ iff $\varsigma_i'\vDash\varphi$. It should be noted that every trace available at a state $s_i$ is also available at $s'_i$, and for every trace available at $s'_i$, there is a trace from $s_i$ with this as a suffix. This means for formulae of the form $E\varphi$, we already have one direction, and where $\varphi$ is prefix-independent (e.g. $E(GF\psi_1\wedge FG\psi_2)$), we have both directions. For $|\varphi|=1$, both $s_i$ and $s'_i$ have the same propositional labellings. Similarly for any path $\varsigma'\in \sigma'_i$, if $\varsigma'\vDash p U q$ we can find a path with the same propositional labelling at each state in $w_i$, otherwise we can find a corresponding path in $\sigma_i$.

    Most of the inductive cases are routine and laborious, so we will just show $s'_i\vDash E(\psi_1 U \psi_2)$ implies $s_i\vDash E(\psi_1 U \psi_2)$ under the IH. Suppose this holds at some $s'_i$. Then, there is some infinite path $\lambda$ from $s'_i$ such that at some $j$, $\lambda[j]\vDash \psi_2$ and for all $k < j$, $\lambda[k]\vDash \psi_1$. Suppose $\psi_2$ is realised after travelling through a path $w_i$; then, this can also be realised from $s_i$ by taking the suffix of the path from this point. So, suppose $\psi_2$ is realised along some path $\varsigma'\in\sigma'_i$. If this is a path such that $\varsigma'\nvDash pUq$, then by IH there is a path $\varsigma\in\sigma_i$ which realises $\psi_2$, and in which $\psi_1$ holds up until that point. Otherwise, there is a path $\rho\in w_i$ which realises $\psi_2$, and in which $\psi_1$ holds up until that point. Both of these paths are available from $s_i$, so $s_i\vDash E(\psi_1 U \psi_2)$ iff $s'_i\vDash E(\psi_1 U \psi_2)$.

    \begin{figure}
    \centering
    \begin{tikzpicture}[node distance=1cm, 
        roundnode/.style={circle, draw, inner sep=3pt},
        dotnode/.style={circle, fill=black, inner sep=3pt},
        cross/.style={path picture={\draw[black] (path picture bounding box.south east) -- (path picture bounding box.north west) (path picture bounding box.south west) -- (path picture bounding box.north east);}},
        crossnode/.style={circle, draw, cross, inner sep=3pt}]
        
        \node [roundnode, label=below:$s_0$] (s0) {};
        \node [roundnode, left=2cm of s0, label=below:$s_1$] (s1) {};
        \node [roundnode, left=2cm of s1, label=below:$s_2$] (s2) {};
    
        \node [roundnode, above=1.2cm of s1, xshift=0.5cm, label=left:$s'_1$] (s'1) {};
        \node [roundnode, above=1.2cm of s2, xshift=0.5cm, label=left:$s'_2$] (s'2) {};

        \node [roundnode, above left=1cm of s1, xshift=0.5cm, label=left:$l$] (notu1) {};
        \node [roundnode, above left=1cm of s2, xshift=0.5cm, label=left:$l$] (notu2) {};
    
        \node [left=1cm of s2] (start) {$\cdots$};
        \node [left=1cm of s'2] (start') {$\cdots$};
    
        \node [left=1cm of s0] (w1entry) {};
        \node [left=1cm of s1] (w2entry) {};
    
        \draw[->] (s'1) to[loop right, out = -45, in = 45, looseness = 10] node[midway, fill=white]{$\sigma'_1$} (s'1);
        \draw[->] (s'2) to[loop right, out = -45, in = 45, looseness = 10] node[midway, fill=white]{$\sigma'_2$} (s'2);
    
        \draw[->] (s1) to[loop right, out = 135, in = 45, looseness = 10]  node[near start, fill=white]{$\sigma_1$} (s1);
        \draw[->] (s2) to[loop right, out = 135, in = 45, looseness = 10]  node[near start, fill=white]{$\sigma_2$} (s2);


        \draw[->] (s1) to[bend left = 400, out = 50] (notu1);
        \draw[->] (s2) to[bend left = 400, out = 50] (notu2);

        \draw[->] (notu1) to[] (s'1);
        \draw[->] (notu2) to[] (s'2);
    
        \draw[in=180,out=-90] (s'2) to node[]{} (w2entry);
        \draw[in=180,out=-90] (s'1) to node[]{} (w1entry);
    
        \draw[->] (start) to (s2);
        \draw[->] (s2) to node[midway, fill=white]{$w_2$} (s1); 
        \draw[->] (s1) to node[midway, fill=white]{$w_1$} (s0); 
        \draw[->] (s0) to[loop right, out = -45, in = 45, looseness = 10] (s0);
    \end{tikzpicture}
    \caption{Model separating $ECTL + E(G\varphi U \psi)$ and $ECTL + E(GF\psi \wedge FG\varphi)$}
    \label{fig:GUseparator}
\end{figure}
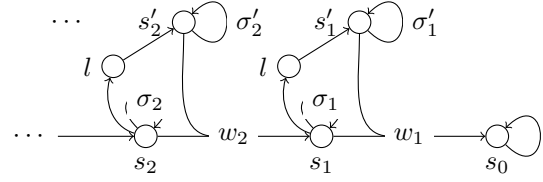
\end{proof}

\ectlplusgu*
\begin{proof}
    This follows from the proof of Theorem 4.3 in \cite{10.1145/567067.567081}, which defines two classes of models $(M_i)_{i\in \mathbb{N}},(N_i)_{i\in \mathbb{N}}$ (Figure \ref{fig:GF-GF-separator}) that are indistinguishable by any ECTL formula but separated by the ECTL$^+$ formula $E(GF\psi \wedge GF \varphi)$. We need to show an extra inductive step, that $N_i,\alpha_i\vDash E(G\psi U \varphi)$ iff $M_i,a_i\vDash E(G\psi U \varphi)$ assuming the same is true of $\varphi$ and $\psi$. Suppose $N_i,\alpha_i\vDash E(G\psi U \varphi)$; then there is some state $\alpha_j$ or $\beta_j$ (call this $\sigma$) for $j < i$ s.t. $N_i,\sigma\vDash\varphi$, for which every state $\sigma'$ between $\alpha_i$ and $\sigma$ is s.t. $N_i,\sigma'\vDash \psi \vee \varphi$. For the analogous states to $\sigma,\sigma'$ in $M_i$ (call these $s$ and $s'$), by IH $M_i,s\vDash \varphi$ and for each $s'$, $M_i,s'\vDash \psi$. So, the path $a_ib_{i}a_{i-1}\ldots s^\omega \vDash G(\psi U \varphi)$ - this path exists from $a_i$ in $M_i$, so $M_i,a_i\vDash E(G\psi U \varphi)$ as required. The other direction can be argued symmetrically.
    \begin{figure}
    \centering
    \begin{tikzpicture}[node distance=1cm, 
    roundnode/.style={circle, draw, inner sep=3pt},
    dotnode/.style={circle, fill=black, inner sep=3pt},
    cross/.style={path picture={\draw[black] (path picture bounding box.south east) -- (path picture bounding box.north west) (path picture bounding box.south west) -- (path picture bounding box.north east);}},
    crossnode/.style={circle, draw, cross, inner sep=3pt}]

\node[dotnode, label=above:$\alpha_1$] (a1) {};
\node[roundnode, below = of a1, label=below:$\beta_1$] (b1) {};

\draw[->] (a1) -- (b1);
\draw[->] (b1) to[bend left] (a1);
\draw[->] (b1) to[loop right] (b1);
\draw[->] (a1) to[loop right] (a1);

\node[left=2cm of a1, anchor=west] (N1label) {$N_1:$};

\node[dotnode, right=4cm of a1, label=above:$a_1$] (m1) {};
\node[roundnode, below = of m1, label=below:$b_1$] (n1) {};

\draw[->] (m1) -- (n1);
\draw[->] (n1) to[loop right] (n1);
\draw[->] (m1) to[loop right] (m1);

\node[left=2cm of m1, anchor=west] (M1label) {$M_1:$};

\node[dotnode, below=3cm of a1, label=above:$\alpha_{i+1}$] (an1) {};
\node[roundnode, below =of an1, label=below right:$\beta_{i+1}$] (bn1) {};

\draw[->] (an1) -- (bn1);
\draw[->] (bn1) to[bend left] (an1);
\draw[->] (bn1) to[loop right] (bn1);
\draw[->] (an1) to[loop right] (an1);

\node[below= of bn1] (Ni) {$M_i$};
\draw[->] (bn1) -- (Ni);

\node[left=2cm of an1, anchor=west] (Nn1label) {$N_{i+1}:$};

\node[dotnode, below=3cm of m1, label=above:$a_{i+1}$] (am1) {};
\node[roundnode, below =of am1, label=below right:$b_{i+1}$] (bm1) {};

\draw[->] (am1) -- (bm1);
\draw[->] (bm1) to[loop right] (bm1);
\draw[->] (am1) to[loop right] (am1);

\node[below= of bm1] (Ni) {$M_i$};
\draw[->] (bm1) -- (Ni);

\node[left=2cm of am1, anchor=west] (Mn1label) {$M_{i+1}:$};

\node[dotnode, right=2.5cm of m1] (dotpropositions) {};
\node[roundnode, below=0.5cm of dotpropositions] (roundpropositions) {};

\node[right of =dotpropositions] {$\vDash p \wedge \neg q$};
\node[right of =roundpropositions] {$\vDash \neg p \wedge q$};

\end{tikzpicture}
    \caption{Recreation from \cite{10.1145/567067.567081}}
    \label{fig:GF-GF-separator}
\end{figure}
\end{proof}

\end{document}